\documentclass{LMCS}

\def\dOi{11(2:4)2015}
\lmcsheading%
{\dOi}
{1--25}
{}
{}
{Apr.~\phantom04, 2014}
{Jun.~\phantom09, 2015}
{}

\ACMCCS{[{\bf Theory of computation}]: Theory and algorithms for
  application domains---Database theory---Database query languages
  (principles)\,/\, Database query processing and optimization (theory)}

\subjclass{H.2}

\usepackage{stmaryrd}
\usepackage{amsthm}
\usepackage{amsmath}
\usepackage{array}
\usepackage[applemac]{inputenc}
\usepackage{bussproofs}
\usepackage{xspace}
\usepackage{multirow}
\usepackage{setspace}
\usepackage{colortbl}
\usepackage{amssymb}
\usepackage{subfigure}
\usepackage{hyperref}
\usepackage{multicol}
\usepackage{tikz}

\usepackage{float}
\usepackage{graphicx}
\usepackage{caption}

\usepackage{subfigure}

\newcommand{\desphi}[1]{desc_{#1}}
\newcommand{\FOCTPd}{\exists^+(child,desc_{\phi})[down](c,s)}
\newcommand{\FOBFKd}{\exists^+(child,desc)[down](c,s)}
\newcommand{\FOCTP}{\exists^+(child,desc_{\phi})}

\newcommand{\untilpath}{ (\downarrow;?\phi)^*;\downarrow}
\newcommand{\upath}[1]{(\downarrow;?#1)^*;\downarrow}

\newcommand{\until}[2]{\langle{#1}^+\rangle{#2}}
\renewcommand{\phi}{\varphi}
\newcommand{\descstar}{\tup{\downarrow^*}}
\newcommand{\child}[1]{\tup{\downarrow}{#1}}
\newcommand{\dEsc}[1]{\tup{\downarrow^+}{#1}}

\newcommand{\mlcont}{\subseteq_{\mathsf{ML}}}

\newcommand {\pvar}{\mathsf{Prop}}

\newcommand {\ch}[1]{\langle{#1}\rangle}

\newcommand {\nat}{\mathbb{N}}

\newcommand {\TP}{\mathsf{TP}}

\newcommand {\CQ}{\mathsf{CQ}}
\newcommand {\CTP}{\mathsf{CTP}}
\newcommand {\CXP}{\mathsf{CTP}}

\newcommand {\ct}{\mathsf{c}}

\newcommand {\RXP}{\mathsf{RXPath}}
\newcommand {\EMXP}{\mathsf{\mu TP}_{\kern -2pt\mathord{\lor}}}

\newcommand {\ECTL}{\exists\textsc{CTL}}
\newcommand {\CTL}{\textsc{ CTL}}

\newcommand {\sem}[1]{\llbracket {#1} \rrbracket}

\newcommand{\tup}[1]{\langle #1\rangle}

\newcommand{\PTIME}{\textsc{PTime}\xspace}
\newcommand{\PSPACE}{\textsc{PSpace}\xspace}
\newcommand{\EXPSPACE}{\textsc{ExpSpace}\xspace}
\newcommand{\EXPTIME}{\textsc{ExpTime}\xspace}

\newcommand{\coNP}{\textsc{coNP}\xspace}

\renewcommand{\phi}{\varphi}

\theoremstyle{plain}

\begin{document}

\title[Containment of CTP]{Containment for Conditional Tree Patterns}

\author[A.~Facchini]{Alessandro Facchini\rsuper a}
\address{{\lsuper A}IDSIA - 
Dalle Molle Institute for Artificial Intelligence}
\email{alessandro.facchini@idsia.ch}

\author[Y.~Hirai]{Yoichi Hirai\rsuper b}
\address{{\lsuper b}FireEye}
\email{yoichi.hirai@fireeye.com}

\author[M.~Marx]{Maarten Marx\rsuper c}
\address{{\lsuper{c,d}}ISLA, University of Amsterdam}
\email{\{maartenmarx,e.sherkhonov\}@uva.nl}
 
\author[E.~Sherkhonov]{Evgeny Sherkhonov\rsuper d}
\address{\vspace{-18 pt}}
\thanks{\rsuper{c,d}This research was supported by the Netherlands Organization for Scientific Research (NWO) under project number 612.001.012 (DEX).}
\keywords{XPath, XML, Conditional Tree Pattern, Containment, Complexity}

\begin{abstract}
\noindent A Conditional Tree Pattern ($\CTP$) expands an XML tree pattern with
 labels attached to the descendant edges.
These labels can be XML element names or Boolean $\CXP$'s.
The meaning of a descendant edge labelled by $A$ and ending in a node
labelled by $B$ is a path of child steps ending in a $B$ node such
that all intermediate nodes are $A$ nodes.
In effect this expresses the \emph{until B, A holds} construction from
temporal logic.

This paper studies the containment problem for $\CTP$.
For tree patterns ($\TP$), this problem is known to be $\coNP$-complete.
We show that it is $\PSPACE$-complete for $\CTP$. This increase in complexity 
is due to the fact that $\CTP$ is expressive enough to encode an unrestricted form of label negation: ${*}\setminus a$, meaning "any node except an $a$-node". Containment of $\TP$ expanded with this type of negation is already $\PSPACE$-hard.

$\CTP$ is a positive, forward, first order fragment of Regular XPath. Unlike $\TP$, $\CTP$ expanded with disjunction is not equivalent to unions of $\CXP$'s.
Like $\TP$, $\CTP$ is a natural fragment to consider:
 $\CTP$ is closed under intersections and  $\CTP$ with disjunction is equally expressive  as positive existential first order logic expanded with the until operator. 

\end{abstract}

\maketitle

\section{Introduction}

Tree Patterns, abbreviated as $\TP$, are one of the most studied languages for XML documents and used in almost all aspects of XML data managment. Tree patterns are a natural language: over trees, unions of tree patterns are equally expressive as positive first order logic \cite{bene:stru05}.
Also, like relational conjunctive queries, the semantics of $\TP$'s can be given by embeddings from patterns to tree models \cite{amer:tree02}.
Equivalence and containment of $\TP$'s is decidable in $\PTIME$ for several fragments \cite{amer:tree02}, and $\coNP$ complete in general \cite{mikl:cont04}.
Tree patterns can be represented as trees, as in Figure~\ref{fig:tp}, or be given a natural XPath-like syntax.

\begin{figure}
 \begin{center}
  \includegraphics[scale=0.8]{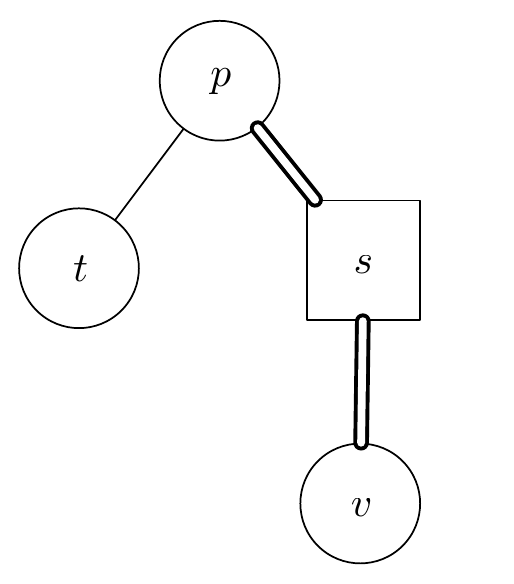}
 \end{center}
  \caption{The tree pattern corresponding to the XPath expression \texttt{/p[t]//s[.//v].} The node in the square box denotes the output node.}
  \label{fig:tp}
\end{figure}

\begin{figure}[tbp]
\begin{center}
\includegraphics[width=0.30\textwidth]{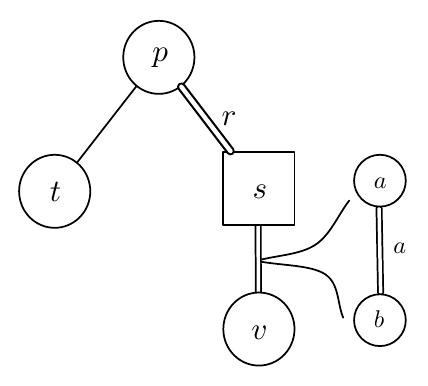}
\caption{Conditional tree pattern corresponding to \eqref{eq:CTP1}. 
 }
\label{fig:CTP1}
\end{center}
\end{figure}

In this paper, we study the expansion of tree patterns with the conditional
descendant axis. 
 We call this expansion \emph{Conditional Tree Patterns}, abbreviated as $\CTP$.
Where the descendant axis in $\TP$ can be written as
the transitive reflexive closure of the XPath step $\mathtt{child{::}*}$, the
conditional descendant axis is the transitive closure of
$\mathtt{child::p[F_1][F_2]\ldots[F_n]}$, where $n\geq 0$ and each $F_i$ is an XPath expression which might contain the conditional descendent axis itself,  followed by a child step.
Syntactically, the expansion is straightforward: in the tree representation add labels  representing conditional tree patterns themselves to the edges.  
Figure~\ref{fig:CTP1} contains an example which is equivalent to the following XPath-like formula in which we use $(\cdot)^*$ to denote the reflexive transitive closure of a path formula:

 \begin{align}
\mathtt{se}&\mathtt{lf::p[child::t]/(child::r)^*/child::s} \nonumber\\
& \mathtt{\ \ [ (child::a[(child::a)^*/child::b])^*/child::v]} \label{eq:CTP1}
\end{align}
The edge in Figure~\ref{fig:CTP1} labeled by $r$ corresponds to the  expression $\mathtt{(child::r)^*/child::*}$, and merely states that all intermediate nodes are labeled by $r$. The other labeled edge shows a nesting of patterns, corresponding to a nested transitive closure statement: all intermediate nodes have to be labeled by $a$ and moreover have to start an $a$-labeled path ending in a $b$-node.

Conditional tree patterns are the forward fragment of Conditional XPath  \cite{marx:cond05} without disjunction and negation. The conditional descendant axis is closely related to the strict until operator from temporal logic \cite{marx:cond05,libk:reas10}. 

\subsection*{Main results}
Our main results concern the expressivity of $\CTP$ and the complexity of the containment problem.
We consider two types of models: the standard XML trees in which each node has exactly one label, and trees in which nodes can have an arbitrary number of labels.  These latter, called \emph{multi-labeled trees}, are the models considered in temporal logic. All our results hold for both semantics.
Models with multiple labels are a convenient logical abstraction for reasoning about tree patterns expanded with attribute value equalities. These are expressions of the form $@_a=c$, where $a$ is an attribute, $c$ a constant, meaning that it holds at a node if and only if the value of the $a$-attribute of the node equals $c$. (With that we can express XPath formulas like \texttt{//table[@border = '1']}).

Containment of tree patterns has been studied in \cite{amer:tree02,mikl:cont04, neve:comp06}. The most relevant results for this paper are that containment of $\TP$'s is $\coNP$-complete in general \cite{mikl:cont04} and $\PSPACE$-complete when the domain of labels is finite \cite{neve:comp06}. 
We show that containment of $\CTP$'s is $\PSPACE$-complete (with both finite and infinite domain of labels). Interestingly, this increase in complexity 
is due to the fact that $\CTP$ is expressive enough to encode an unrestricted form of label negation: ${*}\setminus a$, meaning ``any node except an $a$-node". We show that containment of $\TP$'s expanded with this form of negation is already $\PSPACE$-hard. 
The matching upper bound for   $\CTP$ containment is easily obtained by a translation into Existential $\CTL$ ~\cite{kupf:auto00}.
As a contrast we consider the expansion of $\TP$ with a \emph{safe} form of propositional negation $n\setminus a$~, which selects nodes with the label containing $n$ and not $a$, instead of ${*}\setminus a$ \cite{bara:guar11}. Note that this construct only makes sense on models with multiple labels. 
With respect to expressivity, we show that most results for $\TP$ can be generalized to $\CTP$.
$\CTP$'s can be interpreted in trees by generalizing the $\TP$-embeddings to  simulations known from temporal logic. Similarly to the characterization for $\TP$ in  \cite{bene:stru05}, we show that $\CTP$'s  with disjunction and union are equally expressive as positive first order logic expanded with an \emph{until} operator. From this we obtain that like $\TP$'s, $\CTP$'s are closed under taking intersections.

\subsection*{Organization}
The paper is organized as follows.
This section is continued with a few more comparisons between $\TP$ and $\CTP$, related work and a motivating example.
Section~\ref{sec:prelim} contains preliminaries. Section~\ref{sec:expr} contains the expressivity results and Section~\ref{sec:cont} the results on the complexity of the containment problem.  
We end with conclusions and open questions.

\subsection*{Comparing logical properties of $\TP$ and $\CTP$}

A characteristic difference between $\TP$ and Relational Conjunctive Queries ($\CQ$) is the disjunction
property: if A $\models B\vee C$, then $A \models B$ or $A\models
C$. This holds for $\CQ$, but not for $\TP$. A counterexample is
$//p\models /p\  \vee\  /{*}//p$. 
The languages $\TP$ and $\CTP$ differ on the following:

\begin{description}
\item[Unions] $\TP$ expanded with disjunction is equally expressive as
	   unions of $\TP$~\cite{bene:stru05}. However, $\CTP$ with disjunction is more expressive
	   than unions of $\CTP$.
\item[Countermodels] If containment between two $\TP$'s fails, there is 
	   a countermodel for it of polynomial size~\cite{mikl:cont04}.  Countermodels for
	   $\CXP$ containment may be exponential.
\item[Complexity] $\TP$ containment is $\coNP$-complete~\cite{mikl:cont04} and 
$\CTP$ containment is $\PSPACE$-complete. For both languages this  remains true if we add
disjunction to the language (for $\TP$ see~\cite{neve:comp06}).
\end{description}
However, there are a few useful technical results that $\TP$ and $\CTP$ share.
\begin{description}

\item[Monotonicity] $\TP$ and $\CTP$ formulas are preserved
	   under extensions of  models at the leaves (i.e. when the
          original model is a sub\emph{tree} of the extension). This means that if a $\TP$ ($\CTP$) formula holds in a tree, then it holds in the extensions of the tree. 

\item [Multiple output nodes] The containment problem for $\TP$ and $\CTP$ formulas with multiple output nodes can be reduced to containment of Boolean $\TP$ and $\CTP$ formulas (i.e. when the formula has a single output node which is the root).
\item[Containment for unions] Containment for unions of  $\TP$s can be reduced in $\PTIME$ to checking a set of containments between $\TP$ formulas~\cite{mikl:cont04}. This reduction can be seen as a weak disjunction property.
A similar property holds for $\CTP$ (see Proposition~\ref{prop:mik}).
\item[Multi-labeled models] There is a $\PTIME$ reduction from the containment problem over trees in which each node has exactly one label to  the containment problem over multi-labeled trees. This holds for both $\TP$ and $\CTP$. 
\end{description}

\noindent For $\TP$, most of these results are in~\cite{mikl:cont04}. Here we show how their proofs
can be generalized to $\CTP$. 

\subsection*{Motivation}

Tree patterns exhibit a nice tradeoff between expressive power and the complexity of static analysis. However, there are natural scenarios where tree patterns are not powerful enough, e.g. in Example~\ref{ex:cond}. 
 The conditional axis gives us more querying capabilities while preserving some of the nice properties of tree patterns (see the comparison above).   
 \begin{exa}
 \label{ex:cond}
Conditional tree patterns are  used to query tree shaped structures. As an example, take the tree structure of the UNIX file system. In this file system, every file and directory has different access permissions (read, write or execute) for different type of users (the user, the group and others). Thus, the file system can be modeled as a tree where each node corresponds to a directory or a file (labeled by "dir" and "file" respectively) and has a required attribute for each pair $(user, access\ right)$ which takes values from $\{0,1\}$. A file can only be a leaf in the tree.

Assume we want to ask for all the files that are readable by the user. This means that we are looking for precisely those files for which the following permissions hold:
\begin{itemize}
\item the file is readable for the user,
\item the directory in which the file resides is both readable and executable for the user,
\item the same holds recursively for all the directories from the root to the file.
\end{itemize}
This query can be neatly expressed as a $\CTP$ path formula:
$$\mathtt{/(child::dir[@_{(user,read)}=1][@_{(user,execute)}=1])^*/ child::file[@_{(user, read)}=1]}.$$
Additionally, $\CTP$ can  express non-trivial constraints over the tree representing the file system. For instance, the  formula 
\[
\mathtt{//self::file[@_{(other,read)}=1]} \rightarrow \mathtt{/(child::dir[@_{(other,read)}=1]} \mathtt{[@_{(other,execute)}=1])^*/ *} 
\]
imposes the constraint that for  all files which are  readable it holds that  every directory on the path from the root to this file must be both readable and executable by others.

\end{exa}

 It is known that the conditional axis cannot be expressed by tree patterns or in  Core XPath~\cite{marx:cond05}. On the other hand, the queries from Example~\ref{ex:cond} can be expressed in the positive forward fragment of Regular XPath, which is more expressive than $\CTP$. However, we believe this additional expressive power leads to increase in the complexity (of containment) than for conditional tree patterns. 

\subsection{Related Work}
The complexity of tree patterns is studied in a number of papers and
since \cite{neve:comp06} a virtually complete picture exists for the
complexity of the containment problem for positive fragments of XPath.
Miklau and Suciu~\cite{mikl:cont04} show that the complexity of the containment
problem for $\TP$ with filters, wildcard and descendant is $\coNP$-complete.
Containment and equivalence for fragments of $\TP$ were studied
before. The most interesting result is that containment for $\TP$
without the wildcard is in $\PTIME$ \cite{amer:tree02}. 

Our conditional tree patterns are a first order fragment of conjunctive regular
path queries.  Calvanese et al.~\cite{calv:cont00} show that the
complexity of containment of these queries is $\EXPSPACE$-complete,
but these are interpreted on general graph models.

Conditional XPath \cite{marx:cond05} and conditional tree patterns
are closely related to branching time temporal logic $\CTL$
\cite{clar:auto86}. The conditional child
axis and the strict until connective are interdefinable.
Kupferman and Vardi~\cite{kupf:auto00}
show that the containment  problem for $\ECTL$, which is the
restriction of $\CTL$ to formulas having only the $\exists$ path
quantifier in front of them, is a $\PSPACE$-complete problem. The
positive
$\ECTL$-fragment without until  was also
studied by Miklau and Suciu~\cite{mikl:cont04}. They show that the containment problem for this fragment is equivalent to the  $\TP$-containment problem and thus also $\coNP$-complete. 

This paper studies the containment problem without presence of schema information. A number of complexity results for containment in $\TP$ w.r.t DTDs are given in~\cite{neve:comp06}. In particular, containment for $\TP$ with filters, the wildcard and descendent w.r.t DTD is $\EXPTIME$-complete. This hardness result together with an $\EXPTIME$ upper bound for Conditional XPath~\cite{marx:cond05}, which contains $\CTP$, gives us $\EXPTIME$-completeness of containment for $\CTP$ in the presence of DTDs.

\section{Preliminaries}
\label{sec:prelim}

We review the basic definitions of XML trees, Regular XPath and its semantics.
Then we present Tree Patterns and Conditional Tree Patterns as sublanguages of Regular XPath. 
Tree patterns have an alternative semantics in terms of embeddings
\cite{mikl:cont04}. We give such an ``embedding semantics'' for
Conditional Tree Patterns using ``until-simulations'' in Section~\ref{sec:expr}.

\subsection{Trees}
We work with node-labeled unranked finite trees, where the
node labels are elements of an infinite set $\Sigma$.
Formally, a tree over $\Sigma$ is a tuple $(N, E, r, \rho)$, where
$N$, the set of nodes of the tree, is a prefix closed set of finite
sequences of natural numbers,  $E=$ $\{ (\langle n_1, \dots, n_k\rangle,
\langle n_1, \dots, n_k, n_{k+1}\rangle)$ $~\mid ~$ $ \langle n_1,
\dots, n_{k+1}\rangle \in N\}$ is the edge or child relation, $r$ is the
root of the tree, that is the empty sequence, and $\rho$ is the function
assigning to each node in $N$ a finite subset of $\Sigma$. We refer a tree over $\Sigma$ just as a tree if $\Sigma$ is clear from the context.

Trees in which $\rho(\cdot)$ is always a singleton are called \emph{single-labeled} or\emph{ XML trees}.
Trees without this restriction are called \emph{multi-labeled trees}.

We denote by $E^+$ the descendant relation, which is the transitive
closure of the edge relation $E$, and by $E^*$ the reflexive and transitive
closure of $E$, and by $E(x)$ the set of all children of the node $x$. A node $n$ is a 
\emph{leaf} if $E(n)$ is empty. A \emph{path} from a node $n$ to
a node $m$ is a sequence of nodes $n=n_0, \dots, n_k=m$, with $k>0$,
such that for each $i\leq k$, $n_{i+1} \in E(n_i)$.
We call a \emph{branch}  any maximal path starting from
the root. If $mE^+n$, $m$ is called an \emph{ancestor} of $n$, and if $mEn$, $m$ is called the
\emph{parent} of $n$.

  If $n$ is in $N$, by $T.n$ we denote the subtree of $T$
rooted at $n$.
A \textit{pointed tree} is a pair $T,n$ where $T$ is a tree and $n$ is a
node of $T$.  The \emph{height} of a pointed tree $T,n$ is the length of the
 longest path in $T.n$.

Given two trees $T_1=(N_1,E_1, r_1, \rho_1)$ and $T_2=(N_2, E_2, r_2, \rho_2)$ such that $N_1$ and $N_2$ are disjoint, we define the result of \emph{fusion} of $T_1$ and $T_2$, denoted as $T_1\oplus T_2$, as the tree obtained by joining the trees $T_1$ and $T_2$ without the roots under a new common root labeled by the union of the labels of the roots of $T_1$ and $T_2$. Formally, $T_1\oplus T_2$ is the tree $T=(N,E, r, \rho)$, where $N= (N_1\setminus\{r_1\})\cup (N_2\setminus \{r_2\})\cup \{r\}$, $E=(E_1\setminus \{\tup{r_1,n}\mid n\in N_1\})\cup (E_2\setminus \{\tup{r_2, n}\mid n\in N_2\})\cup \{\tup{r,n}\mid \tup{r_1, n}\in E_1 \text{ or }\tup{r_2, n}\in E_2\}$ and

 $\rho(n) =
  \begin{cases}
   \rho_1(r_1)\cup \rho_2(r_2) & \text{if } n=r, \\
   \rho_1(n)       & \text{if } n\in N_1\setminus \{r_1\}, \\
   \rho_2(n)      & \text{if } n\in N_2\setminus \{r_2\}.
   \end{cases}
$

\subsection{XPath and Tree Patterns}

We define Tree Patterns and Conditional Tree Patterns as sublanguages
of Regular XPath~\cite{cate:expr06}.  

\begin{defi}[Forward Regular XPath]
 Let $\Sigma$ be an infinite domain of labels. Forward Regular
XPath, $\RXP$ for short, consists of node formulas
 $\phi$ and path formulas $\alpha$ which are defined by the following grammar
\begin{center}
\begin{tabular}{ccl}
$\phi$ & $::=$ &  $p \mid \top \mid  \lnot \phi \mid\phi \lor \phi  \mid\phi
 \land \phi \mid \langle \alpha \rangle \phi$ \\
$\alpha$ & $::=$ & $ \downarrow\  \mid\  ?\phi \mid \alpha ; \alpha \mid \alpha
	 \cup \alpha \mid \alpha^* $,  
\end{tabular}
\end{center}
where $p \in \Sigma$.
\end{defi}

We will use $\alpha^+$ as an abbreviation of $\alpha;\alpha^*$.

For  the semantics of $\RXP$,
given a   tree $T=(N, E, r, \rho)$ over $\Sigma$,
the relation $\sem{\alpha}_T\subseteq N\times N$ for a path expression
$\alpha$ and  the satisfaction relation $\models$ between pointed trees
 and node formulas are inductively defined as follows:
 \begin{itemize}
  \item $\sem{\downarrow}_T = E$,
    \item $\sem{ ? \phi }_T = \{(n,n)\in N\times N\mid\, T,n \models \phi\}$,
  \item $\sem{ \alpha;\beta }_T = \sem{\alpha}_T \circ \sem{\beta}_T$, 
    \item $\sem{ \alpha\cup\beta }_T = \sem{\alpha}_T \cup \sem{\beta}_T$,
  \item $\sem{\alpha^\ast}_T=  (\sem{\alpha}_T)^\ast
      $,
 \end{itemize}
and
 \begin{itemize}
 \item $T,n\models \top$,
  \item $T,n\models p$ iff $p\in \rho(n)$,
    \item $T,n\models \lnot\phi$ iff $T,n\not\models\phi$,
  \item $T,n\models \phi\land\psi$ iff $T,n\models\phi$ and $T,n\models\psi$,
  \item $T,n\models \phi\lor \psi$ iff $T,n\models\phi$ or
	$T,n\models\psi$,
  \item $T,n\models \ch{\alpha}\phi$ iff there is a node~$m$ with
	$(n,m)\in \sem{ \alpha}_T$ and $T,m\models\phi$.
 \end{itemize}

 Sometimes we will write $T\models \phi$ to denote $T,r\models \phi$. If the latter holds, we say that $T$ is a \emph{model} of $\phi$.

 \subsubsection*{(Conditional) Tree Patterns}
 \label{sec:cond-tree-patt}

Tree patterns are the conjunctive fragment of $\RXP$ without unions of
paths and with a strongly restricted Kleene star operation. Node
formulas correspond to Boolean tree patterns, and path formulas to
tree patterns with one output node.  

 In the following definitions we again let $\Sigma$ be an infinite domain of labels.
We define Tree Patterns by restricting the syntax of $\RXP$
as follows:

\begin{defi}[Tree Pattern]
Tree Patterns ($\TP$) consist of node formulas $\phi$ and path formulas $\alpha$ defined by the following grammar:
\begin{center}
\begin{tabular}{ccl}
$\phi$ & $::=$ &  $p \mid \top \mid  \phi
 \land \phi \mid \langle \alpha \rangle \phi$ \\
$\alpha$ & $::=$ & $ \downarrow\  \mid\  ?\phi \mid \alpha ; \alpha \mid \ \downarrow^+, $  
\end{tabular}
\end{center}
where $p \in \Sigma$.
\end{defi}

For example, the tree pattern from Figure~\ref{fig:tp} can be written as the path formula $\alpha = ?(p \land \child{t}); \downarrow^+; ?(s\land \child{v})$.
Conditional Tree Patterns are also defined by restricting the syntax of $\RXP$, where we allow conditional descendant paths.
\begin{defi}[Conditional Tree Pattern]
Conditional Tree Patterns ($\CTP$) consist of node formulas $\phi$ and path formulas defined by the following grammar:
\begin{center}
\begin{tabular}{ccl}
$\phi$ & $::=$ &  $p \mid \top \mid  \phi
 \land \phi \mid \langle \alpha \rangle \phi$ \\
$\alpha$ & $::=$ & $ \downarrow\  \mid\  ?\phi \mid \alpha ; \alpha \mid
	 (\downarrow;?\phi)^*;\downarrow\enspace$, 
\end{tabular}
\end{center}
where $p \in \Sigma$.
\end{defi}

The tree from Figure~\ref{fig:CTP1} can be written as the path formula
$$?(p \land \child{t}); \upath{r}; ?(s\land \tup{ \upath{(a \land \tup{\upath{a}}b)}}v ).$$
Note that the node formula  $\langle (\downarrow;?\phi)^*;\downarrow  \rangle \psi$ is
exactly the strict until $\exists U(\psi,\phi)$ from branching time
temporal logic $\CTL$~\cite{kupf:auto00}. 
We will abbreviate this formula simply as $\until{\phi}{\psi}$. 
Because of the equivalences $\langle \alpha;\beta\rangle\phi\equiv \langle\alpha\rangle \langle\beta\rangle\phi$ and $\langle?\phi\rangle\psi\equiv \phi\wedge\psi$, $\CTP$ node formulas can be given the following equivalent definition:
\begin{center}
\begin{tabular}{ccl}
$\phi$ & $::=$ &  $p \mid \top \mid  \phi
 \land \phi \mid \langle \downarrow \rangle \phi \mid \langle \downarrow^+ \rangle \phi  \mid \until{\phi}{\phi}$.
\end{tabular}
\end{center}

Even though the syntax is slightly different, $\CXP$ is the
conjunctive forward only fragment of Conditional XPath~\cite{marx:cond05}.

\subsubsection*{Expansions}

We consider two expansions of $\TP$ and $\CTP$,
 with negated labels and with disjunction in node formulas together with union in paths. 
Negated labels is a restricted type of negation where only the construct $\neg p$ is allowed in the node formulas. 
We denote expansions of the language $L$ with one or
two of these features by $L^S$ for $S\subseteq \{\neg,\vee \}$. 

\subsubsection*{Query evaluation}
In \cite{gott:effi05} it was shown that the query evaluation problem for Core XPath is $\PTIME$-complete in the combined complexity. As noted in \cite{marx:cond05}, using results on model checking for Propositional Dynamic Logic, this can be extended to Regular XPath, and thus to all our defined fragments. 

\begin{fact}[\cite{marx:cond05}]
Let $T$ be a tree, $n_1, n_2$ nodes in $T$, and  $\alpha$ a Regular XPath path formula. The problem whether $(n_1, n_2)\in \sem{\alpha}_T$ is decidable in time $O(|T|\times |\alpha|)$ with $|T|$ the size of the tree and $|\alpha|$ the size of the formula.
\end{fact}

\subsection{Containment}
As we are considering two kinds of expressions, path and node
expressions, we have  different notions of containment.

\begin{defi} Let $\varphi$ and $\psi$ be two $\RXP$-node
 formulas. We say that
$\varphi$ is \emph{contained} in $\psi$, notation  $\phi \subseteq \psi$,
if  for every $T,n$,
 $T,n\models\varphi$ implies $T,n\models\psi$.
 Let $\alpha$ and $\beta$ be two $\RXP$-path
 formulas. We say that
\begin{itemize}
\item $\alpha$ is \emph{contained as a binary query} in $\beta$, denoted
      $\alpha \subseteq^2 \beta$,
if for any
  tree $T$,
 $\sem{\alpha}_T \subseteq \sem{\beta}_T$,
 \item $\alpha$ is \emph{contained as a unary query} in $\beta$, denoted
       \mbox{$\alpha \subseteq^1 \beta$},
if for any
  tree $T$ with root $r$, and any node $n$,
 $(r, n) \in \sem{\alpha}_T$ implies  $(r,n) \in\sem{\beta}_T$\enspace.
\end{itemize}
Containment over single-labeled trees is denoted by $\subseteq$, and containment over multi-labeled models by $\mlcont$. 
\end{defi}

Luckily, these three notions are closely related, and containment of
path formulas can be reduced to  containment of node formulas and vice versa
(cf. also \cite{neve:comp06,mikl:cont04}).

\begin{prop}\label{prop:output}
Let $\alpha$ and $\beta$ be two $\RXP$-path formulas, $\phi$ and $\psi$ $\RXP$-node formulas, in which negation is restricted to labels only.
\begin{enumerate}[label=(\roman*)]
\item  $\alpha
 \subseteq^2 \beta$ iff $\alpha \subseteq^1 \beta$,

\item
Let  $p$ be a
 label not occurring in $\alpha$ or in $\beta$. Then,
 $\alpha \subseteq^2 \beta$ iff $\langle \alpha \rangle\langle \downarrow \rangle p \subseteq
 \langle \beta \rangle\langle \downarrow \rangle p$,
 
\item $\phi\subseteq \psi$ iff $?\phi \subseteq^2 ?\psi$.
 
\end{enumerate} 
All the above items also hold for the case of multi-labeled trees.  
\end{prop}

 \begin{proof}
 (i) ($\Rightarrow$) Let $T=(N, E, r, \rho)$ be a tree such that $(r, n)\in \sem{\alpha}_T$. By the assumption, we have $\sem{\alpha}_T\subseteq \sem{\beta}_T$. This implies that $(r, n)\in \sem{\beta}_T$, which was required to show. 
 
 ($\Leftarrow$) Let  $T=(N, E, r, \rho)$ be a tree, $n_1$ and $n_2$ in $N$ such that $(n_1, n_2)\in \sem{\alpha}_T$. 
  Since $\alpha$ is from the forward fragment of Regular XPath, we have that $(n_1, n_2)\in \sem{\alpha}_{T.n_1}$. Using the assumption, we have $(n_1, n_2)\in \sem{\beta}_{T.n_1}$. Then by monotonicity, we obtain $(n_1, n_2)\in \sem{\beta}_T$, which was desired. 
 
 (ii) ($\Rightarrow$)
 Suppose $\alpha \subseteq^2 \beta$, and let $T, n \models \tup{\alpha}\tup{\downarrow} p$. Thus there is a descendant $m$ of $n$ in $T$ such that
  $(n,m) \in \sem{\alpha}_T$ and $T, m \models \tup{\downarrow}p$. By the hypothesis we
  therefore have $(n,m) \in \sem{\beta}_T$ and thus $T, n \models
  \tup{\beta}\tup{\downarrow}p$.

 ($\Leftarrow$) Assume $\tup{\alpha}
  \tup{\downarrow} p \subseteq \tup{\beta}
  \tup{\downarrow} p$. Consider a tree $T=(N, E, r, \rho)$ and
  a pair $(n,m) \in \sem{\alpha}_T$. Since $p$ does not occur in
  $\alpha$, we may assume that $T, n\not\models p$ for every node $n\in N$. We then define the
  tree $T'=(N', E', r, \rho')$, where 
  \begin{itemize}
  \item $N':=N\cup \{x\}$,
  \item $E':=E'\cup \{\tup{m,x}\}$,
  \item $\rho'(y):= \left\{ 
  \begin{array}{l l}
   p & \quad \text{if\ } y= x,\\
   \rho(y) & \quad \text{otherwise.}
  \end{array} \right.
$
  \end{itemize}
 
By definition of $T'$, we have $T', m\models \tup{\downarrow} p$ and $(n, m)\in \sem{\alpha}_{T'}$.  
Hence, $T', n\models \tup{\alpha}\tup{\downarrow}p$. By the assumption, it implies that $T', n\models \tup{\beta}\tup{\downarrow}p$.  
Since $p$ holds at $x$ only, and $x$ is a child of $m$, we have that $(n, m)\in \sem{\beta}_{T'}$.  By definition of $T'$, the path between $n$ and $m$ is also in $T$. Thus, $(n, m)\in \sem{\beta}_T$, as desired.  

Item (iii) easily follows from the definitions. 
 \end{proof}

\noindent In light of Proposition~\ref{prop:output}, from now on we consider the containment problem of node formulas only. 
We now give an interesting example of $\CTP$ containment.

\begin{exa}
Let us consider the $\CXP$ node formulas
\begin{eqnarray*}
  \label{eq:2}
 \phi&=&\until{(\until{b}{d})}{\until{a}{b}} \ \text{and}\\
 \psi &=&\until{(\until{c}{d})}{\until{a}{b}}.
\end{eqnarray*}
 Although it is hard to see
 from the first glance, $\phi\subseteq \psi$ holds.
 Indeed, in every
 model $T$ of $\phi$ either there is a direct child of the root where
 $\until{a}{b}$ holds, or there is a path $r=v_1,\ldots, v_n$ in $T$ to
 the node $v_n$ where $\until{a}{b}$ holds and $\until{b}{d}$ holds in
 every node $v_i$ ($1< i< n$) on the path. In the first 
 case, $\psi$ holds at the root since $\until{a}{b}$ holds at the direct
 child.  

 Now let's consider the second case. Let $j\in \{2,\ldots,n-1\}$ be
 the least number with the property that $v_j$ has a non-empty $b$-path to the $d$-descendant. If there is no such number,
 then $\until{c}{d}$ holds at every $v_i, 1< i<n$ (as each of them
 has a direct $d$-child) and, thus, $T,r\models \psi$.  

Assume there is such a $j$. Since $v_j$ has a $b$-node as a child, we have that $T,v_j\models \until{a}{b}$. Moreover, $\until{c}{d}$ holds at each $v_i$ for  $i< j$ since $v_i$ has a $d$-node as a child.
 Thus in this case we obtain that $T,r\models \psi$ holds too.   
\end{exa}

\newcommand{\down}{\downarrow}

\section{Expressivity}
\label{sec:expr}

We extend the semantics of $\TP$ given by embeddings of queries into trees to $\CTP$. Instead of embeddings we need a simulation known from temporal logic.

\subsection{Interpreting Conditional Tree Patterns by simulations}
The semantics of conditional tree patterns can be defined using 
simulations developed for LTL \cite{blac:moda01}. 
These simulations generalize the embeddings for tree patterns from
\cite{amer:tree02,mikl:cont04} with an additional clause for checking the
labels on the edges.

We start with defining the tree pattern analogues of $\CTP$ node and
path expressions.
\begin{defi}
A  \emph{conditional tree pattern} is a 
 finite tree $(N, E , r, \bar{o}, \rho_N, \rho_E)$ with labeled nodes
 and edges, where $N$ is  the set of nodes of
 the tree, 
 $E\subseteq N\times N$ is the set of  edges, 
 $r$ is the root of the tree, $\bar{o}$ is a $k$-tuple ($k> 0$) of output nodes,
 $\rho_N$   is the function assigning to each node in $N$ a finite set of
 labels from $\Sigma$ 
 and $\rho_E$ is
 the function assigning to each pair in $E$ either $\downarrow$ or a
 Boolean conditional tree pattern.
\\
A \emph{Boolean conditional tree pattern} is a conditional tree
pattern with a single output node which equals to the root.
\\
A conditional tree patterns is said to have \emph{multiple output nodes} if the number $k=|\bar{o}| $ is greater than 1.
\\ A \emph{tree pattern} is a conditional tree pattern whose edges are
only labeled by $\downarrow$ or $\top$.
\end{defi}

To be  consistent with the pictorial representation of $\TP$, in
$\CTP$ patterns an edge labeled with $\downarrow$ is drawn as a single
line, while an edge labeled with a $\CTP$ node formula is drawn as a
double line with a $\CTP$ pattern as the label (e.g.  as in
Figure~\ref{fig:CTP1}). The output nodes have the square shape.

$\CTP$ node and path expressions can be translated into (Boolean)
conditional tree patterns with one output node and vice-versa. The translations are given in Appendix~\ref{app:tra}.
We denote the equivalent (Boolean) conditional tree pattern of  a
$\CTP$ path or node  expression $\alpha$ or $\phi$ by 
 $\ct(\alpha)$ and $\ct(\phi)$, respectively. 

Next we generalize the notion of $\TP$-embeddings to $\CTP$-simulations. 
\begin{defi}
\label{def:sim}
Let $T=(N, E, r, \bar{o},\rho_N,
 \rho_E)$ be a conditional tree pattern as in the previous definition and $T'=(N', E',
 r', \rho')$ a tree. A total function $f: N \to N'$ 
 is called a \emph{simulation} from the pattern $T$  into the pointed tree
 $T',r'$ if it satisfies the following properties:
\begin{description}
\item[root preserving]  $f(r)=r'$;
\item[label preserving] if   $p \in \rho_N(n)$, then $p \in \rho'(f(n))$;
\item[child edge preserving] if  
  $nEn'$ and $\rho_E(n,n')= '\downarrow'$, then $f(n)Ef(n')$;

\item[conditional edge simulation] if  
  $nEn'$ and $\rho_E(n,n')$ is not equal to $\downarrow$, then $f(n)E'^+f(n')$ and 
        for every $x$ such that
      $f(n)E'^+x E'^+f(n')$ there is a simulation  from the Boolean
      conditional tree pattern $\rho_E(n,n')$
      into  $T'.x$ (the subtree of $T'$ rooted at $x$).
\end{description}
\end{defi}
When the pattern is a tree pattern, $\downarrow$ and $\top$ are the
only labels on edges. For the label $\top$, the 'conditional edge
simulation' clause trivializes to checking that $f$ is an embedding
for the descendant edges. Simulations for tree patterns are thus
equivalent to tree pattern embeddings \cite{amer:tree02,mikl:cont04}.

The next theorem states that simulations can be used to evaluate
conditional tree patterns.

\begin{thm}\label{prop:simc}
Let $\phi$ and $\alpha$ be a $\CXP$ 
 node and path expression, respectively. Let  $T$ be a
 tree and $n$ a node in $T$.
 \begin{enumerate}[label=(\roman*)]
\item $T,n
 \models \phi$ if and only if there is a simulation from $\ct(\phi)$ into $T.n$.

\item $(n,n')\in \sem{\alpha}_T$ if and only if there is a simulation
from $\ct(\alpha)$ into $T.n$ which relates the output node of
$\ct(\alpha)$ to node $n'$.
\item  Items (i) and (ii)  also hold when $T$ is an infinite tree.
\end{enumerate}
\end{thm}

The proof is by mutual induction on the node and  path expressions.
\begin{rem}
\label{rem1}
  The notions of (conditional) tree pattern and embedding for tree
  patterns and simulation for conditional tree patterns are easily
  extended to work for expansions of these languages with negated
  labels (denoted by $\TP^{\neg}$ and $\CTP^{\neg}$, respectively). 
  
  For (conditional) tree patterns, add a second node labelling
  function $\rho_N^{\neg}(\cdot)$ which also assigns each node a finite set
  of labels. These are interpreted as the labels which are false at
  the node.

  For the embeddings and simulations, in Definition~\ref{def:sim} we add a clause stating that also
  negated labels are preserved: 
  \begin{description}
  \item[negated label preserving] If $ p\in \rho_N^{\neg}(n)$, then $p\not\in \rho'(f(n))$.
  \end{description}

  With these modifications, Theorem~\ref{prop:simc} also holds for
  $\TP^{\neg}$ and   $\CTP^{\neg}$. 
\end{rem}


\subsection{Expressivity characterization}

In this section we give expressivity characterization for $\CTP$ similar to the one for various fragments of XPath in \cite{bene:stru05}. The exact logical characterization allows one to compare different fragments of (Regular) XPath and derive non-trivial closure properties such as closure under intersection. 
We show that $\CTP$ path formulas correspond to a natural fragment of first order logic ($\mathbf{FO}$) and are closed under intersections. 

For $\phi$ a formula, let $\desphi{\phi}(x,y)$ be an abbreviation of the "until" formula $desc(x,y) \wedge \forall z(desc(x,z)\wedge desc(z,y)\rightarrow \phi(z))$. 
Let $\FOCTP$ be the fragment of first order logic built up from the binary relations $child$ and  $desc$, label predicates $p(x)$ for each label $p\in \Sigma$ and  equality $'{=}'$, by closing under $\wedge,\ \vee$ and $\exists$ as well as under the rule:
\begin{eqnarray}
  \text{if } \phi(x)\in \FOCTP, \text{ then } 
 \desphi{\phi}(x,y) \in \FOCTP. \label{crule}
\end{eqnarray}
We use $\FOCTP(c,s)$ to denote $\FOCTP$ formulas with exactly two
variables $c$ and $s$ free.   
Following \cite{bene:stru05}
we restrict the $\mathbf{FO}$ fragment to its downward fragment $\FOCTPd$
by requiring that every bound variable as well as $s$ is syntactically
restricted to be a descendant of $c$ or equal to $c$.

Note that $\FOCTPd$ without the rule \eqref{crule} is the logic \linebreak
$\FOBFKd$ from  \cite{bene:stru05}, and shown there to be
equivalent to unions of tree patterns.
Even though $\FOCTPd$ does not contain negation, not every formula is satisfiable (e.g. $desc(x,x)$). 

We can now characterize  conditional tree patterns in terms of
$ \FOCTPd$.

\begin{thm}
  \label{thm:exp1}
   The following languages are equivalent in expressive power:
  \begin{itemize}
  \item unions of the false symbol and $\CTP$ paths  which can have  disjunctions in the node formulas
  \item $ \FOCTPd$.
  \end{itemize}
\end{thm}

\begin{proof}
A standard translation $TR_{xy}(\cdot)$ turns any $\CTP$ path formula $\alpha$ into an equivalent formula in $\FOCTPd$. The translation is essentially the semantics of path and node formulas written in the first order language. 

\begin{center}
\begin{tabular}{l c l}
$TR_{xy}(\emptyset)$& = & $child(x,y)\land x=y$ \\
$TR_{xy}(\downarrow)$& =  & $child(x,y)$ \\
$TR_{xy}(?\phi)$ & = & $ x=y \land TR_x(\phi)$ \\
$TR_{xy}(\alpha_1; \alpha_2)$ &  =  &$\exists z. (TR_{xz}(\alpha_1)\land TR_{zy}(\alpha_2))$, \\
& & where $z$ is a fresh variable \\
$TR_{xy}(\untilpath)$ & = & $\desphi{TR_z(\phi)}(x,y)$  \\
$TR_{xy}(\alpha_1\cup \alpha_2)$ & = & $TR_{xy}(\alpha_1)\lor TR_{xy}(\alpha_2)$ \\
& & \\
$TR_{x}(p)$& =  & $p(x)$ \\
$TR_{x}(\top)$ & = & $x=x$ \\
$TR_{x}(\phi_1\land \phi_2)$ &  =  &$TR_x(\phi_1)\land TR_x(\phi_2)$ \\
$TR_{x}(\phi_1\lor \phi_2)$ & = & $TR_x(\phi_1) \lor TR_x(\phi_2)$ \\
$TR_{x}(\tup{\alpha}\phi)$ & = & $\exists y. (TR_{xy}(\alpha)\land TR_y(\phi))$, \\
& & where $y$ is a fresh variable.
\end{tabular} 
\end{center}

By definition,   $TR_{cs}(\alpha)$ is a formula in $ \FOCTPd$ and equivalent to $\alpha$. 

For the other direction, let $\theta\in \FOCTPd$. We use our earlier notation $\desphi{\psi}(x,y)$ for the "until" formulas. First we introduce two special formulas $\top(x)$ and $\bot(x)$ which stand for $x=x$ and $desc(x,x)$, respectively.  

We will modify $\theta$ in several  steps. First we replace $child(x,y)$ and $desc(x,y)$ by the equivalent $\desphi{\bot}(x,y)$ and $\desphi{\top}(x,y)$, respectively. Then eliminate all equalities by renaming variables. 
Bring all existential quantifiers inside $\theta$ to the front,  bring the body into disjunctive normal form and distribute the disjunctions over the quantifiers. We then end up with a disjunction of formulas of the form $\exists\bar{x}\phi(c,s)$, with $\phi$ a conjunction of formulas $\desphi{\psi}(x,y)$, $p(x)$, $\top(x)$ and $\bot(x)$. 
If $\phi$ contains $\bot$ replace $\exists\bar{x}\phi(c,s)$ with $\bot$.  

As our target language is closed under unions and $\bot$, we only need to translate the "conjunctive queries"  $\exists\bar{x}\phi(c,s)$.   
Consider the graph of the variables in $\phi$ in which two variables $x,y$ are related if $\phi$ contains an atom $\desphi{\psi}(x,y)$. If the graph is cyclic, it cannot be satisfied on a tree and $\exists\bar{x}\phi(c,s)$ is equivalent to $\bot$. If it is a tree, $c$ will be the root.  Assuming that we can rewrite all $\psi$ in the atoms $\desphi{\psi}(x,y)$, we can rewrite it as a conditional tree pattern.
If $\psi$ is a boolean combination of $p(z)$, $\top(z)$ and $\bot(z)$ atoms, this is not hard:   Simply bring it into disjunctive normal form, and remove each conjunct containing $\bot$. If the result is not empty,    then translate to a union of (trivial) conditional tree patterns. If it is empty, translate $\desphi{\psi}(x,y)$ to a child step.
If $\psi$ contains subformulas of the form  $\desphi{\phi}(x,y)$, we  apply the current  procedure   to it.

Thus assume that the variable graph is a directed acyclic graph. We  eliminate undirected cycles step by step. The length of the cycle equals the number of variables in it. Consider that the graph contains two paths $\pi_1$ and $\pi_2$ both going from $x$ to $y$ and without other common variables. 
We show that this subgraph is equivalent to a union of subgraphs with no or smaller cycles.  

In the simplest case, both paths are of length 1 and thus consist of a $\desphi{\psi}(x,y)$ atom. But then we can use the following equivalence to remove the cycle:
\begin{equation}
 \desphi{\phi}(x,y)\wedge \desphi{\psi}(x,y) \equiv \desphi{\psi\wedge\psi}(x,y)
\end{equation}
So assume the cycle looks like Figure~\ref{fig:paths}.

\begin{figure}[H]
  \centering
  \includegraphics[scale=1.2]{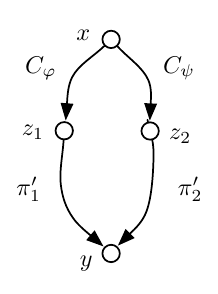}
  \caption{undirected cycle in $\phi$} 
  \label{fig:paths}
\end{figure}

 If this is satisfied on a tree $(N,E, r,\rho)$ with assignment $g$, there are three possibilities: $g(z_1)=g(z_2)$, $g(z_1)E^+ g(z_2)$ or $g(z_2)E^+g(z_1)$. In each case our original formula is equivalent to  one with a smaller cycle. 
The three possibilities are 
\begin{description}
\item[when $g(z_1)=g(z_2)$] $\desphi{\phi\wedge\psi}(x,z_1) \wedge z_1\pi_1'y \wedge z_1\pi_2'y$ 
\item[when $g(z_1)E^+ g(z_2)$]  $ \desphi{\phi\wedge\psi}(x,z_1) \wedge \psi(z_1) \wedge z_1\pi_1'y \wedge \desphi{\psi}(z_1,z_2)\wedge z_2\pi_2'y$
\item[when $g(z_2)E^+g(z_1)$]  $ \desphi{\phi\wedge\psi}(x,z_2) \wedge \phi(z_2) \wedge z_2\pi_2'y \wedge \desphi{\phi}(z_2,z_1)\wedge z_1\pi_1'y$.
\end{description}
In the first case the length of the  cycle decreased by two, in the two other cases by one.

Thus their disjunction is equivalent to the  formula of Figure~\ref{fig:paths}. 
 We replace that formula by this disjunction, bring the result  in disjunctive normal form and distribute the disjuncts out. We again have a disjunction of  "conjunctive queries". As the new cycles are smaller, this procedure will terminate, and results in  a (big) disjunction of trees.
 \end{proof}

An important consequence of this result is that
$\CTP$ patterns are closed under intersection:

\begin{thm}
  \label{thm:exp3}
  The intersection of two $\CTP$ paths is equivalent to $\bot$ or a union of $\CTP$ paths. 
\end{thm}

Theorem~\ref{thm:exp1} together with the translation of Regular XPath
into  $\mathbf{FO}^*(c,s)$ from \cite{cate:expr06} implies that every union of $\CTP$ patterns with
disjunctions is in the intersection of first order logic and  positive
downward $\mathbf{FO}^*(c,s)$. It is an intriguing open problem whether the
converse also holds.

\section{Containment}
\label{sec:cont}
Before we determine the exact complexity of  the containment problem for $\CTP$ and expansions we prove a number of reductions. Most are generalizations from $\TP$ to $\CTP$. The key new result is the encoding of negated labels in $\CTP$.

\subsection{Containment Preliminaries}

The following reductions will be used later in our upper and lower bound proofs.

\subsubsection*{Multiple output nodes}
\label{sec:mult-outp-nodes}

The main difference between tree patterns and their XPath formulation
is that tree patterns can have multiple output
nodes.
Kimelfeld and Sagiv~\cite{kime:revi08} (Proposition~5.2) show that for the $\TP$ containment problem the number of output nodes is not important: the problem can be $\PTIME$ reduced to a containment problem of Boolean $\TP$s. 
 This is achieved, given $\phi, \psi\in \TP$, by adding a child labeled with a new label $a_i$ to every output node $X_i$ in both $\phi$ and $\psi$. Second, to every leaf that is not an output node or a newly added node, we add a child labeled with $\top$. 
 The same result, using the same argument,  holds for $\CTP$.  

\begin{prop}
\label{prop:boolean}
Let $S\subseteq\{\vee, \neg\}$. For $\CTP^S$ patterns with multiple output nodes $\varphi, \psi$  there are $\PTIME$ computable Boolean $\CTP^S$, $\phi', \psi'$ such that 
$\phi\subseteq \psi\ \mbox{iff}\ \phi'\subseteq \psi'.$
The same holds  for multi-labeled trees. 
\end{prop}

\subsubsection*{Disjunctions in the consequent}

We now show that the containment problem
for both unions of $\CTP^{\neg}$ and unions of $\TP^{\neg}$ can be
reduced to containments without unions. 
This is useful in lower bound proofs, as we can use a union in the
consequent to express multiple constraints.

The proof of the following proposition is a
slight modification of the proof  for $\TP$ in
\cite{mikl:cont04}. 
\begin{center}
\begin{figure}[t]
\begin{tabular}{p{0.20\textwidth}p{0.20\textwidth}}
\includegraphics[width=0.20\textwidth]{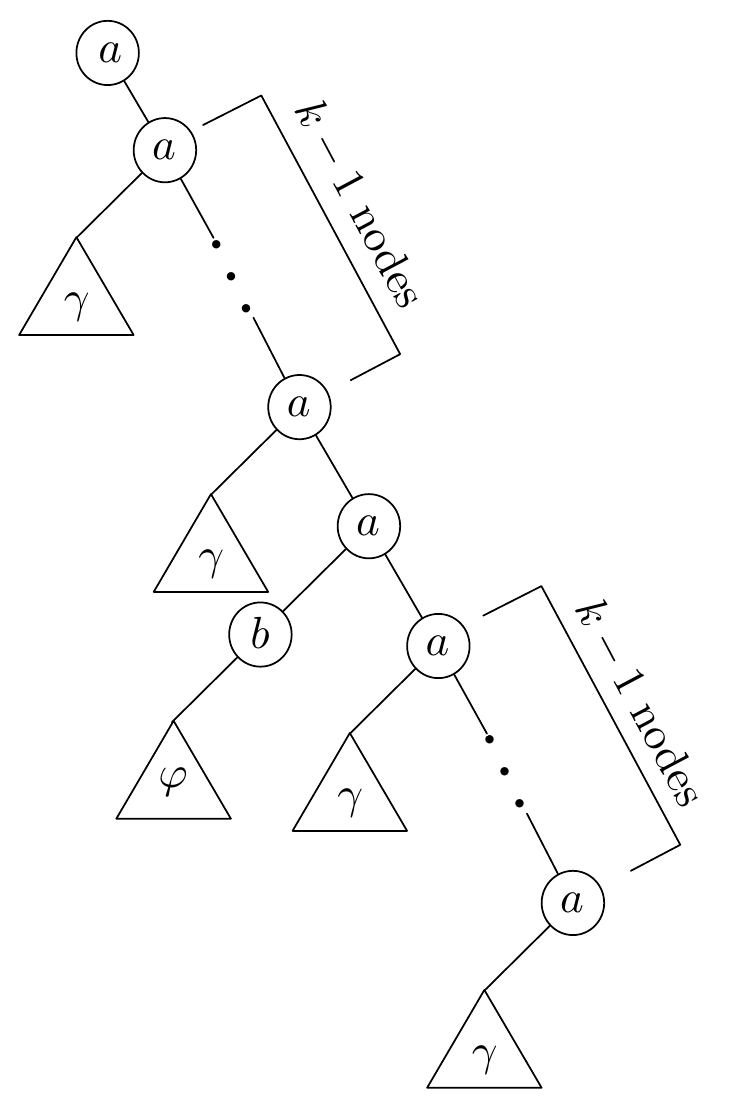}
 & 
\includegraphics[width=0.20\textwidth]{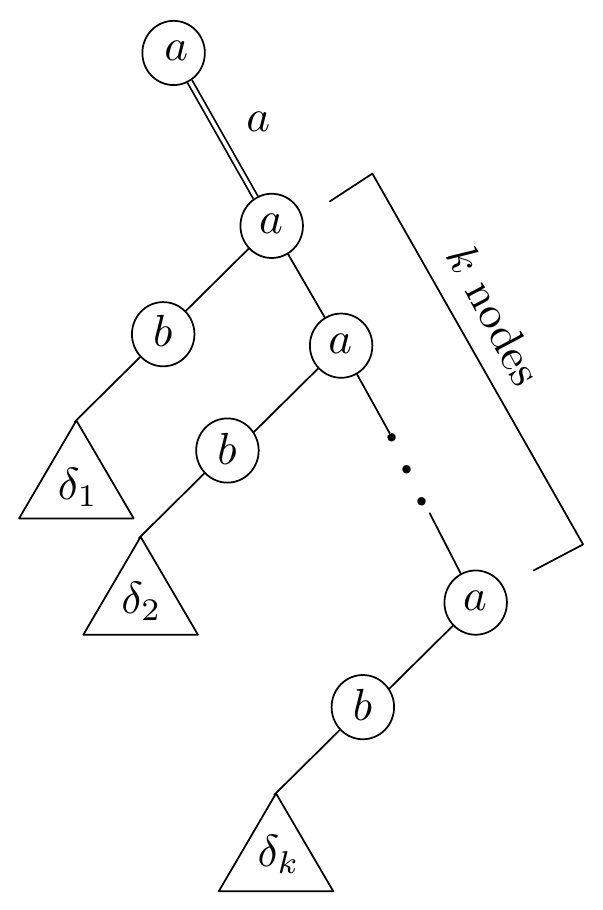}
\\
 \centering
 \ref{fig12}(a) Pattern $\phi'$ 
&
     \centering
     \ref{fig12}(b) Pattern $\psi$
     \\
\end{tabular} 
\caption{Patterns $\phi'$ and $\psi$ from Proposition~\ref{prop:mik}.}
 \label{fig12}
\end{figure}
\end{center}
\begin{figure}[!t]
\begin{center}
\includegraphics[width=0.25\textwidth]{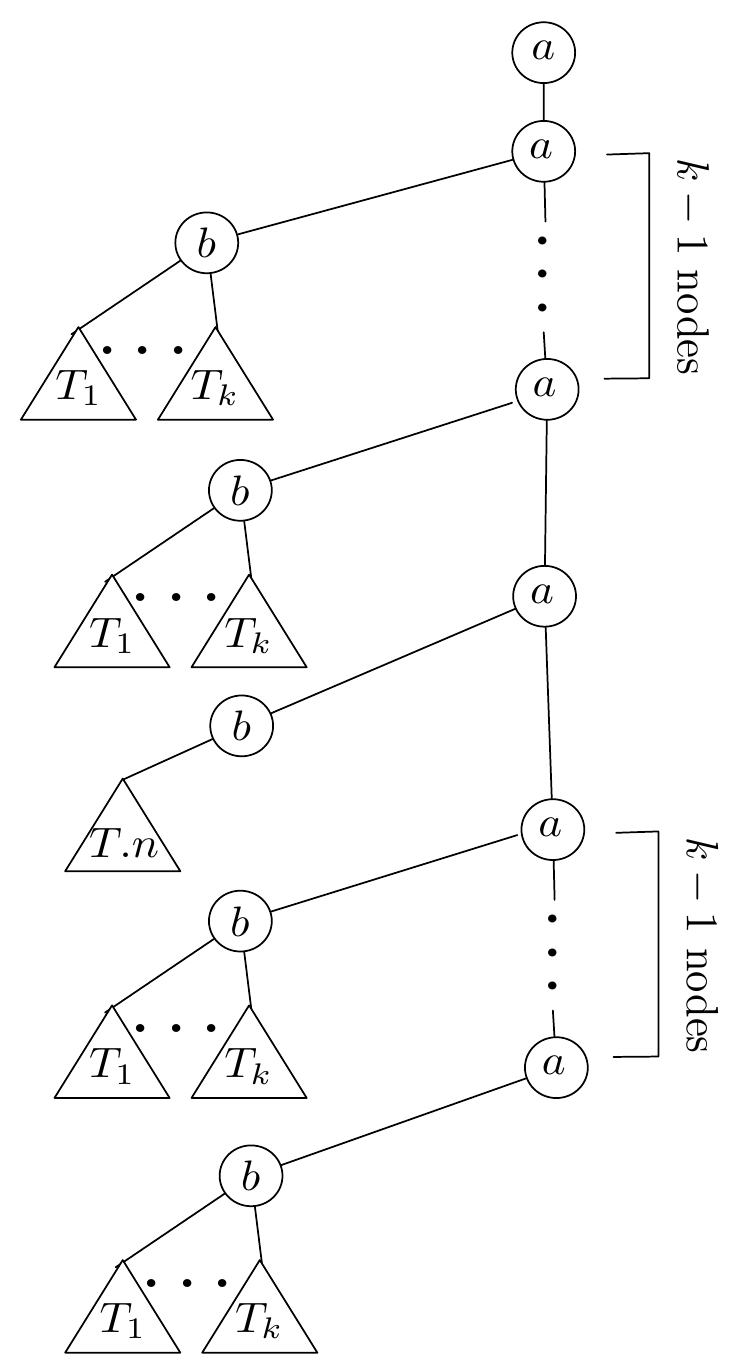}
\caption{Model for pattern $\phi'$  from Proposition~\ref{prop:mik}.}
\label{fig:3}
\end{center}
\end{figure}
\begin{prop}\label{prop:mik}
Let $\phi$ be a $\CXP^{\neg}$ ($\TP^{\neg}$) formula and  $\Delta$ a finite set of
 $\CXP^{\neg}$ (resp. $\TP^{\neg}$) formulas. Then there are $\PTIME$ computable $\CXP^{\neg}$ ($\TP^{\neg}$) formulas $\phi'$ and $\psi$ such that
$\phi \subseteq \bigvee \Delta$ iff 
$\phi' \subseteq \psi$.
\\
If $\phi$ and $\Delta$ are in $\CTP$ ($\TP$),  $\phi'$ and $\psi$ are in $\CTP$ ($\TP$) as well.
\\
The same holds  for  containment over multi-labeled trees. 
\end{prop}
\begin{proof} 
First, for the case of $\TP^{\neg}$, the proof from \cite{mikl:cont04} can be readily applied here (using embeddings which also preserve negated labels, cf Remark~\ref{rem1}).

Now we prove the proposition for $\CTP^{\neg}$. 
For simplicity, we use the same letter to denote
 a $\CXP^{\neg}$ formula and its corresponding tree pattern representation.
Assume $\Delta$ is $\{\delta_1, \dots, \delta_k\}$. If $\iota$ is a
 $\CXP^{\neg}$ pattern, by $\iota^a$ we denote the $\CTP^{\neg}$ pattern
 defined
 as having the root labeled by $a$, whose unique child is the root of
 $\iota$. 
Consider  new labels $a$ and $b$ occurring neither in $\phi$
 nor $\Delta$.
Let   $\gamma=\bigwedge_{0<\ell \leq k}
 \delta^b_\ell$. Note that  $\gamma$ is consistent if every $\delta_\ell$ is consistent.
 Now define  $\phi'$ as the $\CXP^{\neg}$ tree pattern  in
 Figure~\ref{fig12}(a), and  $\psi$ as the $\CXP^{\neg}$ tree pattern  in
 Figure~\ref{fig12}(b).

We verify that $\phi \subseteq \bigvee \Delta$ iff $\phi' \subseteq \psi$. Without loss of generality, we assume that each of $\phi, \delta_1, \dots, \delta_k$ is consistent.

For the direction from left to right, assume $\phi \subseteq \bigvee \Delta$ and $T, n \models \phi'$ for an arbitrary $T,n$. This
 means that there is a path $n=m_0..... m_{2k-1}$ of nodes in $T$ which are labeled with $a$
 and such that:
\begin{itemize}
\item for every $0< \ell \leq k$: $T, m_i \models
      (\delta_\ell^b)^a$, with $i \in \{1, \dots, k-1\} \cup \{k+1,
      \dots, 2k-1\}$,
\item $T, m_k \models (\phi^b)^a$.
\end{itemize}
By the assumption, it follows that there is an index $j$ such that $T, m_k \models (\delta_j^b)^a$. 
Then, by Theorem~\ref{prop:simc}, there is a simulation $f'$ from $(\delta_j^b)^a$ into  $ T.m_k$. 
The simulation $f'$  can be extended to a simulation $f$ from $\psi$
into $T,n$ in the obvious way such that $f$ is a simulation from
$(\delta^b_i)^a$ into $ T,m_{k-j+i}$ and the $a$-descendant edge in
$\psi$ is simulated on the path $m_0, \ldots m_{k-j+1}$. Thus,
$T,n\models \psi$. 

For the other direction, assume  $\phi' \subseteq \psi$ and  $T,n
\models \varphi$ for an arbitrary $T$ rooted at $n$. 
As all $ \delta_\ell$ are consistent, there are trees  $T_\ell,
 n_\ell \models \delta_\ell$, for each $\ell \in \{1, \dots, k\}$. 
Let $T'$ be the tree created from $T.n$ and the $T_\ell$'s depicted in Figure~\ref{fig:3} with root $r$.
Then  $T', r \models \phi'$.
  Because we assumed $\phi' \subseteq \psi$, then also 
 $T', r \models \psi$,  which by Theorem~\ref{prop:simc} implies the
 existence of a simulation $g$
from $\psi$ into $ T'.r$. In particular, the span of $k$ $a$-nodes in
$\psi$ has to be simulated on an $a$-path of length $k$ in $T'$. This
means that for some $j \in \{1, \dots,
k\}$, $g$ is a simulation from $\delta_j$ into $T.n$.  By  
 Theorem~\ref{prop:simc} again, $T, n \models \delta_j$, and
 thus $T, n \models \bigvee \Delta$.
 
 Note that if $\phi$ and $\Delta$ are in  $\CTP$, then the constructed $\phi'$ and $\psi$ are in $\CTP$ too.
\end{proof}
  
\subsubsection*{From single-labeled to multi-labeled trees}  

We now show that the containment problem over single-labeled trees can be reduced to a containment problem over multi-labeled trees. 

\begin{prop}
  \label{prop:standard}
Let $S\subseteq \{\neg,\vee\}$. Given $\CXP^S$  ($\TP^S$) patterns $\phi,\psi$, there are $\PTIME$ computable $\CTP^S$  ($\TP^S$)
patterns $\phi',\psi'$  such that $\phi\subseteq\psi$
iff
$\phi'\mlcont\psi'$.
\end{prop}

\begin{proof}

Let $p_1, \ldots, p_n$ be the labels from $\Sigma$ occurring in $\phi$ or $\psi$.
We take $\phi''$ as $\phi$, and $\psi'':=\psi \vee\bigvee_{1\leq i<j\leq n} (p_i\wedge p_j)  \vee\bigvee_{1\leq i<j\leq n} \until{\downarrow}{(p_i\wedge p_j)}$. The disjuncts added to $\psi$ ensure that every node of a counterexample is labeled with at most one label. 
Now we show that $\phi\subseteq\psi$
iff
$\phi''\mlcont\psi''$

($\Rightarrow$)
We show the contraposition. Let $T$ be a multi-labeled tree such that $T\models
 \phi''$ and $T\not\models \psi''$. Then let $T'$ be the tree obtained from
 $T$ by restricting the labels $\rho'(v):=\rho(v)\cap \{p_1,\ldots ,
 p_n\}$ for every node $v$. Since $T$ does not satisfy $\psi''$, the
 label of every node in $T'$ contains at most one symbol. Because we did not change the valuation of the labels $p_1, \ldots, p_n$,  for each formula $\theta$ constructed from  these tags, it holds that $T,v\models \theta$ iff $T',v\models \theta$. From this and the fact that the label of each node of $T'$ is of size at most one, it follows that $T'\models \phi''$ and $T'\not \models \psi''$.
  In order to make a single-labeled tree out of $T'$\kern -2pt, we add a dummy symbol $q$ in the label of a node if its label in $T'$ is empty. For the resulting tree $T''$ it holds $T''\models \phi$ and $T''\not \models \psi$.

($\Leftarrow$)
Let $T$ be a single-labeled tree such that $T\models \phi$ and $T\not\models \psi$. Then this tree can be considered as a multi-labeled tree. Moreover, $T\models \phi''$ and $T\not\models \psi''$ since $T$ does not satisfy any of the disjuncts in $\psi''$.

 Our transformation does not introduce negations, but only disjunctions in $\psi''$. 
 Thus the proposition is proven for the cases when $S$ contains disjunction, where we take $\phi':=\phi''$ and $\psi':=\psi''$. For the remaining cases ($S\subseteq \{\neg\}$), 
 although $\psi''$ contains disjunctions,
 by Proposition~\ref{prop:mik}, there exist two $\PTIME$ computable $\CTP^S$ ($\TP^S$) patterns
 $\phi'$ and $\psi'$ such that $\phi''\mlcont\psi''$
iff
$\phi'\mlcont\psi'$. This concludes the proof.
\end{proof}
\subsubsection*{Negated labels}
\label{sec:negated-labels}

We now show that $\CTP$ is expressive enough to encode label negation, as far as the containment problem is concerned.

The trick of the encoding lies in the fact that two nodes connected by a descendent edge are in either child or descendant of a child relation. Additionally, if we require that neither of these relations occur at the same time, we can faithfully encode label negation.
\begin{prop}
 \label{prop:neg}
 Let $\phi$ and $\psi$ be $\CTP^\neg$ patterns. There are $\PTIME$ computable $\CTP$ patterns $\phi'$ and $\psi'$ such that 
\[ \phi\subseteq\psi\ \mbox{iff}\ \phi'\subseteq \psi'. \]
This also holds for containment over multi-labeled trees. 
\end{prop}
\begin{proof}
Given a containment problem $\phi\subseteq\psi$ in $\CXP^\neg$,
we  construct an equivalent containment problem
$\phi^\bullet\subseteq\psi^\circ$ with $\phi^\bullet$ in $\CTP$ and
$\psi^\circ$ a union of $\CTP$ patterns.
 Applying Proposition~\ref{prop:mik} then  yields the desired
result.

Let $p_1$, $\neg p_1,\dots,p_n,\neg p_n$ be the labels  appearing in
$\phi$ and $\psi$  and their
negations.
Let $s$ be a new label  and
$\zeta$ be the formula $\ch\downarrow
 (s\land \bigwedge_i \until{\downarrow}{p_i})$.
We define the translation $(\cdot)^\bullet$ inductively:
\begin{align*}
 \top^\bullet
   &= \zeta\\
 (p_i)^\bullet
   &= \zeta \land \ch\downarrow(s\land\ch\downarrow p_i)\\
 (\neg p_i)^\bullet
   &= \zeta\land\ch\downarrow(s\land\ch\downarrow\until{\downarrow}{p_i})\\
 (\theta\land\sigma)^\bullet
   &= \theta^\bullet\land\sigma^\bullet \\
 (\ch\downarrow\phi)^\bullet
 &= \zeta\land\ch\downarrow\phi^\bullet \\
 (\until\phi\psi)^\bullet     &= \zeta\land\until{\phi^\bullet}{\psi^\bullet}.
\end{align*}
The translation is defined in such a way  that for any $\phi$, $\phi^\bullet$ implies $\zeta$. 
Obviously $(\cdot)^\bullet$ is a $\PTIME$ translation.
Let $\descstar \phi$ denote $\phi\lor \until\downarrow\phi$. Let $AX$ be the disjunction of the following formulas:

\begin{eqnarray}
\descstar(s\land \descstar \zeta) \label{ax:1} \\
\bigvee_i \descstar( p_i^\bullet\land(\neg p_i)^\bullet) \label{ax:2}\\
\bigvee_{i\neq j}\descstar(p_i^\bullet \wedge p_j^\bullet) \label{ax:3}.
\end{eqnarray}

In our intended models, none of these disjuncts are true at the root. The  disjunct \eqref{ax:1} is technical, \eqref{ax:2} states that no node makes the encodings of both $p$ and $\neg p$ true, and \eqref{ax:3} ensures that no node makes the encoding of two different labels true. 
We define $\psi^\circ = \psi^\bullet\lor AX$. 
In case of multi-labeled trees the translation is  the same, except that we do not include the disjunction~\eqref{ax:3} in the definition of $\psi^\circ$.
Note that although in $\psi^\circ$ we have disjunctions within the scope
of a modality, we can (in $\PTIME$) rewrite the formula into an equivalent union of disjunction free formulas.

We now show that $\phi\subseteq\psi$ \mbox{iff}
$\phi^\bullet\subseteq\psi^\circ$.
\\
 ($\Leftarrow$)
 We show the contrapositive.
 Given a single-labeled tree~$T=(N, E, r, l)$ such that $T,r\models \phi$ and $T, r\not\models \psi$,
 we build a single-labeled tree~$T'=(N', E', r', l')$ as follows.
 The set of nodes $N'$ is the maximal $E'$-connected subset of $N\times\{0,1,2,3\}\times \{p_1, \ldots, p_n, \sharp\}$ which contains the root $r'=(r,0,\sharp)$. The relation
  $E'$ is defined as follows: 
 $(n',0,\sharp)$ is the parent of $(n,0,\sharp)$ when $n'$ is the parent of $n$ in $T$;
 $(n,0,\sharp)$ is the parent of $(n,1,\sharp)$; $(n,1, \sharp)$ is the parent of $(n,2,p_i), i=1, \ldots, n$, and $(n, 2, p_i)$ is the parent of $(n,3,p_i), i=1, \ldots, n$.  
 The labeling $l'$ is defined as follows: $s$ only labels  nodes of the form $(n,1,\sharp)$. 
  Then,  $p_i$ labels $(n,2,p_i)$ if $T,n\models p_i$ and $p_i$
  labels $(n,3,p_i)$ if
 $T,n\not\models p_i$.  No other nodes are labeled by $p_i$. 
 In all other cases, we label a node by a fresh label $z$. 
 
  It is clear from the definition   that $T'$ is a single-labeled tree.
Note that $T'\not\models AX$ and that $\zeta$ is true   at all nodes of type $(n,0,\sharp)$, and only at these nodes.
Let $\theta$ be a formula in variables  $\{p_1, \ldots, p_n\}$, and $n\in N$. By induction on $\theta$, we show
 $T,n\models\theta$ iff $T',(n,0,\sharp)\models \theta^\bullet$.
 
 \begin{itemize}
 \item $\theta=\top$. This holds because by construction  $\zeta$ is true at all nodes of type $(n, 0, \sharp)$. 
 
 \item $\theta=p$. We have $T, n\models p$ $\Leftrightarrow$ $l(n)=p$ $\Leftrightarrow$ $ l'(n, 2, p)=p$ $\Leftrightarrow$(since $(n, 2, p)$ is the only child of $(n, 1, \sharp)$ labeled with $p$ and $(n, 1, \sharp)$ is labeled by $s$) \\
  $T', (n, 1, \sharp)\models s\land \child{p}$ $\Leftrightarrow$ (since $(n, 1, \sharp)$ is the only child of $(n,0, \sharp)$ and $\zeta$ is always true at $(n, 0, \sharp)$) $T', (n, 0, \sharp)\models \zeta\land \child{(s\land \child{p})}$ $\Leftrightarrow$ $T', (n,0, \sharp)\models p^\bullet$. 
 
 \item $\theta=\neg p$. We have $T, n \models \neg p$ $\Leftrightarrow$ $l(n)\neq p$ $\Leftrightarrow$ $l'(n, 3, p)=p$ $\Leftrightarrow$ (since $(n, 3, p)$ is the only descendent of $(n, 1, \sharp)$ labeled with $p$ and $(n, 1, \sharp)$ is labeled by $s$) \\
  $T', (n, 1, \sharp)\models s\land \child{\dEsc{p}}$ $\Leftrightarrow$ (since $(n, 1, \sharp)$ is the only child of $(n,0, \sharp)$ and $\zeta$ is always true at $(n, 0, \sharp)$) $T', (n,0,\sharp)\models \zeta\land \child{(s\land \child{\dEsc{p}})}$ $\Leftrightarrow$ $T', (n, 0 , \sharp)\models (\neg p)^\bullet$.
 
 \item $\theta=\phi_1\land \phi_2$. We have $T, n \models \phi_1\land \phi_2$ $\Leftrightarrow$ $T, n\models \phi_1$ and $T, n \models \phi_2$ $\Leftrightarrow$ (by the induction hypothesis) $T', (n,0, \sharp)\models \phi^\bullet_1$ and $T', (n, 0, \sharp)\models \phi^\bullet_2$ $\Leftrightarrow$ $T', (n, 0, \sharp)\models \phi^\bullet_1 \land \phi^\bullet_2$ $\Leftrightarrow$ $T', (n,0, \sharp)\models (\phi_1\land\phi_2)^\bullet$.
 
  \item $\theta=\child{\phi}$. We only show the right to left direction. $T', (n, 0, \sharp)\models (\child{\phi})^\bullet$ $\Leftrightarrow$ $T', (n, 0, \sharp)\models  \zeta\land \child{\phi^\bullet}$ $\Leftrightarrow$  $\exists  m \in N'$ such that $(n, 0, \sharp)E' m$ and $T',m \models \phi^\bullet$. But then $T',m \models \zeta$ because $\phi^\bullet$ implies $\zeta$. Thus $m$ must be of the form $(n',0,\sharp)$ as only these make $\zeta$ true. 
  But then $n'\in N$ and  $nEn'$ and we may apply the  inductive  hypothesis to get $T,n'\models\phi$, and thus $T,n\models \child{\phi}$.

  \item $\theta=\until \phi\psi$.  We only show the right to left direction.  $T',(n, 0, \sharp)\models (\until\phi\psi)^\bullet$ $ \Leftrightarrow$ $T', (n, 0, \sharp)\models \zeta \land \until{\phi^\bullet}{ \psi^\bullet}$ 
  $ \Leftrightarrow$
  $\exists m \in N'$ such that $(n,0, \sharp)E'^+m$ and $T', m\models \psi^\bullet$ and $\forall m'. (n,0, \sharp)E'^+ m' E'^+ m$ it holds $T', m'\models \phi^\bullet$. Now, because $\phi^\bullet$ implies $\zeta$ for all $\phi$, $m$ and all nodes $m'$ are of the form $(n',0,\sharp)$ and thus all in the original model $N$. Moreover they stand in the same way in the $E$ relation. Thus we can apply the  inductive hypothesis and obtain  $T,n\models\until \phi\psi$.
 \end{itemize}
 As a special case, $T',(r,0,\sharp)$ satisfies $\phi^\bullet$ but not
 $\psi^\bullet$.
 Recall that  $T',(r,0,\sharp)$ does not satisfy the  other disjuncts $AX$ of
 $\psi^\circ$ either and thus, $T', (r, 0, \sharp)\models \phi^\bullet$ and $T' ,(r, 0, \sharp)\not\models \psi^\circ$, as desired.

 ($\Rightarrow$)
 Again we show the contrapositive.
 Suppose there is a model~$T=(N, E, r, l)$ satisfying $\phi^\bullet$ but not
 $\psi^\circ$ at the root $r$. 
 Then in particular $T,r\not\models AX$.
 Without loss of generality, we can assume that the simulation from
 (the conditional tree pattern corresponding to) $\phi^\bullet$ into $T$
 is surjective. 
 (Otherwise the image of $\phi^\bullet$ is a subtree of $T$ and thus by monotonicity $\psi^\circ$ cannot be satisfied at the root of this subtree.)
 In this model, as a consequence of the fact that $\phi^\bullet$ implies $\zeta$, every branch has an initial segment satisfying~$\zeta$,
 immediately
 followed by a node labeled by $s$, which is followed by a segment
 where $\zeta$ is never satisfied because of the first disjunct \eqref{ax:1} of $AX$. 
 
We define a tree~$T'=(N', E', r', l')$ whose set of nodes $N'$ consists of the nodes
 of $T$ where $\zeta$ is satisfied, i.e.  $N' = \{n\in N\mid T,n\models \zeta\}$;
 $E'$ is simply the restriction of $E$ to $N'$ and  $r'=r$.
 We define $l'(n)= p_i$ iff $T,n\models
 p_i^\bullet$, and if $T, n\not\models p_i^\bullet$ for any $i$, then we label $n$ with a fresh variable $z$.
 This definition of $l'$ is well-defined because the falsity of the disjunction~\eqref{ax:3}  in $AX$ ensures that each node makes at most one $p_i^\bullet$ true in $T$.
 
Let $\theta$ be a formula in variables  $\{p_1, \ldots, p_n\}$, and $n\in N$. By induction on $\theta$, we show
   that 
 \begin{equation}
\label{eq:ctp}
 T,n\models\theta^\bullet\mbox{ iff }   T',n\models\theta.
 \end{equation}

 \begin{itemize}
 \item $\theta=\top$. Then $T, n\models \top^\bullet$ $\Leftrightarrow$ (by definition of $(\cdot)^\bullet$) $T, n \models \zeta$ $\Leftrightarrow$ (by definition of $N'$) $\Leftrightarrow$ $T', n \models \top$. 
 
 \item $\theta=p$. Then $T, n \models p^\bullet$ $\Leftrightarrow$    (by definition of the labeling $l'$)   $l'(n)= p$ $\Leftrightarrow$ $T',n \models p$.

 \item $\theta=\neg p$.  If $T, n\models (\neg p)^\bullet$, then by the falsity of $AX$,  $T, n\not\models p^\bullet$, and thus $l'(n)\neq p$ and $T',n\not\models p$ and thus $T',n\models \neg p$.
 Conversely, $T',n\models \neg p$ $\Leftrightarrow$ $l'(n)\neq p$ $\Leftrightarrow$ $T, n\not\models p^\bullet$. But $T,n\models \zeta$ and thus either $p^\bullet$ or $(\neg p)^\bullet$ must hold at $T,n$. Thus 
 $T, n\models (\neg p)^\bullet$.
 
 \item $\theta=\phi_1\land \phi_2$. Then $T, n\models (\phi_1\land \phi_2)^\bullet$ $\Leftrightarrow$ $T, n \models \phi_1^\bullet$ and $T, n \models \phi_2^\bullet$ $\Leftrightarrow$ ($n\in N'$ because $\phi^\bullet$ implies $\zeta$ and thus by the inductive hypothesis) $T', n\models \phi_1$ and $T', n\models \phi_2$ $\Leftrightarrow$ $T', n\models \phi_1\land \phi_2$. 
 
 \item $\theta=\child \phi$. Then $T,n \models (\child \phi)^\bullet$ $\Leftrightarrow$ $T, n\models \zeta$ and $T, n \models \child \phi^\bullet$ $\Leftrightarrow$   there exists $n'\in N$ such that $nEn'$ and $T, n'\models \phi^\bullet$ $\Leftrightarrow$ (by  the fact $T, n'\models \zeta$ and, thus, $n'\in N'$ and the inductive hypothesis) there exists $n'\in N'$ such that $nE'n'$ and $T', n'\models \phi$ $\Leftrightarrow$ $T', n\models \child \phi$, as desired.

\item $\theta=\until\theta\tau$.
  Assume
 $T,n\models(\until\theta\tau)^\bullet$.
Then, by definition of $(\cdot)^\bullet$,  $T,n\models
 \zeta\land\until{\theta^\bullet}{\tau^\bullet}$.
 That means there exists $n'$ in $T$ with $nE^+n'$ such that $T,n'\models\tau^\bullet$ and for all $n''$ with $nE^+n''E^{+}n'$ it holds $T,n''\models \theta^\bullet$.
 By definition of $(\cdot)^\bullet$, the nodes $n,n'$ and all the nodes between $n$ and $n'$ 
 satisfy $\zeta$ and thus belong to $T'$.
 By inductive hypothesis, $T',n'\models \tau$ and $T',n''\models \theta$ for all $nE^{+}n''E^{+}n'$, which means
 $T',n\models \until\theta\tau$ holds.  Conversely, assume
 $T',n\models \until\theta\tau$. By definition of $T'$, we have that
 $T,n\models \zeta$, which is the first conjunct of
 $({\until\theta\tau})^\bullet$. The second conjunct follows by  the inductive
 hypothesis.  
  \end{itemize}
As  $T$ is a counterexample for
 $\phi^\bullet\subseteq\psi^\circ$,  \eqref{eq:ctp} implies that 
  $T'$ is a counterexample of $\phi\subseteq \psi$, as desired. 
  
The same argument applies for multi-labeled trees. The only change is to remove the last  disjunction~\eqref{ax:3} from $AX$.
\end{proof}

\subsection{Lower bounds}

In this section we show that the containment for $\TP^{\neg}$ is $\PSPACE$ hard. This lower bound will carry over to the containment problem for $\CTP$.

\begin{thm} \hfill
\begin{enumerate}[label=(\roman*)]
\item The containment problem for $\CTP$ is $\PSPACE$-hard.
\item The containment problem for $\TP^{\neg}$ is $\PSPACE$-hard,
\item Both results also hold for multi-labeled trees.
\end{enumerate}
\end{thm}

\begin{proof}
(i) follows from (ii) by Proposition~\ref{prop:neg}. 
 (iii) follows from (i) and (ii) by Proposition~\ref{prop:standard}.
For proving (ii), we reduce the corridor tiling problem, which
is known to be hard for $\PSPACE$~\cite{chle:domi86,emde:conv96}, to the containment
problem for $\TP^{\neg}$. We use the construction from the  $\PSPACE$-hardness proof for the containment problem of $\TP$ with disjunction over a finite alphabet in \cite{neve:comp06}.

The corridor tiling problem is formalized as follows. Let ${\sf
Til}=(D,H,V, \bar{b}, \bar{t}, n)$ be a tiling system, where $D=\{d_1,
\ldots,$ $d_m\}$ is a finite set of tiles, $H, V\subseteq D^2$ are
horizontal and vertical constraints, $n$ is a natural number given in
unary notation, $\bar{b}$ and $\bar{t}$ are tuples over $D$ of length
$n$. Intuitively, given a corridor of width $n$, the goal is to
construct a tiling of the corridor using the tiles from $D$ so that the
horizontal and vertical constraints are satisfied and  the bottom
and  top rows are  tiled by $\bar{b}$ and $\bar{t}$, respectively. We say that a tiling satisfies the horizontal constraints $H$(respectively the vertical constraints $V$) if for every $1\leq i\leq n, j\in \nat$ it holds that if $d_1$ and $d_2$ are the tiles on the positions $(i , j)$ and $(i+1, j)$ of the corridor (respectively $(i,j)$ and $(i, j+1)$), then $(d_1, d_2)\in H$ ($(d_1, d_2)\in V$ ). 

Now we construct in $\PTIME$ in the length of {\sf Til}, two $\TP^{\neg}$ expressions $\varphi$ and $\psi$ such
that the following holds: 
\begin{equation}
\label{eq:tile}
\varphi \not\subseteq \psi\ \mbox{iff  there exists
a tiling for}\ {\sf Til}.
\end{equation}
By  Proposition~\ref{prop:mik}, we may without loss of generality construct $\psi$ as a disjunction of $\TP^{\neg}$ expressions. 
To this purpose, we use
the string representation of a tiling. Each row of the considered tiling is represented by the tiles it
consists of. If the tiling of a corridor of width $n$  has $k$ rows, it
is represented by its rows separated by the special symbol $\sharp$. The
top row is followed by the symbol~$\$$.  Thus,  a tiling is a
word of the form $u_1 \sharp u_2 \sharp \cdots \sharp u_k \$$,
where each $u_i$ is the word of length $n$ corresponding to the $i$-th
row in the tiling. In particular $u_1=\bar{b}$ and $u_k=\bar{t}$.
For the sake of readability, for expression $r$, $r^i$ denotes the path formula $?r;\downarrow; ?r; \ldots; \downarrow; ?r$ with $i$ occurrences of $r$.

Let $\varphi$ be
$$\tup{?b_1;\downarrow; ?b_2 ; \ldots ; \downarrow; ?b_n; \downarrow ; ?\sharp;\  \downarrow^+\ ; ?t_1,\downarrow; \ldots \downarrow ; ?t_n;\downarrow}\$.$$
Intuitively, this
 expression enforces a tiling to start with a path starting with $\bar{b}$ and finishing
 with $\bar{t}$ and the final symbol~$\$$.
 Now the formula $\psi$ defines all incorrect tilings and additional
constraints. It is the disjunction of the following $\TP^{\neg}$ formulas.

\begin{enumerate}
\item[(0)] $\tup{\downarrow^+\ ; ?(\neg d_1\wedge \ldots \wedge \neg d_m\wedge \neg \sharp);\ \downarrow^+}\$$. There is a position that is labeled neither by a tile nor a delimiter.

\item[(1)] Incorrect length of a row.
\begin{enumerate}[label=(1\alph*)]
\item[(1a)] $\bigvee^{n-1}_{i=0}\tup{\downarrow^+\ ;?\sharp; \top^i;\ \downarrow}\sharp$, a row is too short;

\item[(1b)] $\tup{\downarrow^+\ ;(\neg \sharp)^{n+1}} \top$, a row is too long;
\end{enumerate}

\item[(2)] $\bigvee_{(d_1,d_2)\not\in H}\tup{\downarrow^+; ?d_1; \downarrow ; ?d_2} \top$, some horizontal
	   constraints are violated;

\item[(3)] $\bigvee_{(d_1,d_2)\not \in V}\tup{\downarrow^+; ?d_1; \downarrow; \top^n; \downarrow; ?d_2}\top$, some
	   vertical constraints are violated.
\end{enumerate}

\noindent
Note that negated labels are used in (0) and (1b). Also note that the size of $\varphi$ and $\psi$ is bounded by a polynomial in the size of  {\sf Til}.

We now show \eqref{eq:tile}.
 
$(\Leftarrow)$. Assume that there exists a tiling of the corridor. Let
 $s$ be the string representation of it. Then,
 $s=u_1\sharp u_2 \sharp  \dots \sharp u_k \$$, where
 $|u_i|=n$, $u_i \in D^n$, $u_1=\bar b$ and $u_k= \bar t$. Moreover, on
 the one hand  if $x\cdot y$, is an infix of some $u_i$, then $(x,y) \in
 H$, and on the other hand for every infix $x\cdot u' \cdot y$ of length
 $n+1$ of $u_i \sharp \cdot u_{i+1}$, it holds that $(x,y) \in V$. Let $T_s$ be the corresponding tree, i.e. a single path of $|s|$ nodes $\{v_1, \ldots , v_{|s|}\}$ where the labeling is set in accordance with $s$, i.e. $l(v_i)=s_i$.
 Clearly, $T_s$ is a model of $\varphi$ and not of $\psi$.

$(\Rightarrow)$. Let $T$ be a tree such that $T, r\models \varphi$ and
 $T, r\not\models \psi$. Since $T, r\models \varphi$, there must exist a
 path $r=v_1,\ldots, v_m$ in $T$ which starts with $\bar{b}$ and finishes with
 $\bar{t}\$$. Moreover, either $\sharp$ or a symbol from $D$ is the label of every node $v_i, 1\leq i< m$, according to (0). 
 
We define a tiling function $g: \{0, \ldots, n-1\}\times \nat\to D$ assigning a tile to every position in the corridor as follows: $g(i,j)=l(v_{(n+1)\times
 j+i+1}), 1\leq i\leq n$, where $l$ is the labeling function of $T$. Indeed, this function is well defined, as (1) ensures the correct counting. 
By formulas (2) and (3) the tiling defined by $g$ satisfies the horizontal and vertical constraints. 

\end{proof}

The   difference between  $\CTP$ patterns and  $\TP$ patterns is that $\CTP$ descendent edges can be labeled by   patterns. One might ask whether a bound on the degree of such a labeling nesting can lead to a lower  complexity of the containment problem. Define the until nesting depth $un: \CTP \to  \mathbb N$ as follows.
\begin{itemize}
\item $un(p)=un(\top)=0$,
\item $un(\tup{\downarrow}\phi)=un(\until{\downarrow}{\phi})=1$,
\item $un(\phi_1\land \phi_2)=\max\{un(\phi_1),un(\phi_2)\}$, 
\item  $un(\until{\phi}{\psi})=un(\phi)+1$.
\end{itemize}

Unfortunately, a close examination of the encoding in the lower bound proof and the encodings in Propositions~\ref{prop:mik} and~\ref{prop:neg} gives a negative answer to the question. Thus we obtain

\begin{thm}
The $\CTP$ containment problem for formulas of until nesting depth one is $\PSPACE$-hard. 
\end{thm}

\subsection{Upper bounds}
\label{sec:ctp}

In this section, we show a matching $\PSPACE$ upper bound for $\CTP$ containment.

The complexity of $\CTP^{\neg,\vee}$ containment follows from a translation into existential
$\CTL$ ($\ECTL$),
whose  containment problem is known to be $\PSPACE$-complete
\cite{kupf:auto00}.
The only small technical issue is that $\ECTL$ formulas are
interpreted over \emph{infinite} finitely branching trees. We solve that by relativizing formulas  with a new
propositional variable $s$, whose interpretation will provide the desired finite tree.

\begin{thm}
  \label{thm:ctp-up}
The containment problem for $\CXP^{\neg,\vee}$ is decidable in $\PSPACE$.
This also holds for multi-labeled trees. 
\end{thm}
\begin{proof}
We prove the upper bound for the case of multi-labeled trees. The upper bound for single-labeled trees then follows by Proposition~\ref{prop:standard}. 
 
We recall from \cite{kupf:auto00} the definition of $\ECTL$. Let $\pvar$ be a set of propositional variables.
$\ECTL$ node formulas $\phi$ and path formulas $\alpha$ are
defined by:

\begin{align*}
\phi &::=   \top \mid \bot \mid p \mid \lnot p \mid \phi \land \phi \mid \phi
\lor \phi  \mid \exists \alpha
 \\
\alpha &::=   \phi \mid  \alpha \land \alpha \mid \alpha \lor \alpha \mid X
\alpha \mid  (\alpha U \alpha) \mid  (\alpha \hat{U} \alpha)
\end{align*}
where $p \in \pvar$, and asking that in node formulas, the
operators $X$(``next"), $U$(``until") and $\hat{U}$(``dual of until") are
always immediately preceded by $\exists$.

The semantics for $\ECTL$ is given by \emph{infinite finitely
branching trees}. $\ECTL$ path formulas are interpreted at a (possibly infinite) paths of the tree and  $\ECTL$ node formulas are interpreted at nodes of the tree. 
Our reduction will translate $\CTP^{\neg , \vee}$ formulas to $\ECTL$ node formulas, with occurrences of formulas of the form $\exists X \phi$ and $\exists \psi U \phi$. Semantics of such formulas is defined as follows, given an infinite tree $T=(N, E, r, \rho)$ and a node $n\in N$. 

\begin{align*}
T, n\models \exists X\phi\  & \text{iff} \  \text{there exists a node }m\in N \ \text{such that}\ nEm\ \text{and}\ T, m\models \phi, \\
T, n\models \exists \psi U \phi\ & \text{iff}\  \text{there exists a node}\ m\in N \ \text{such that}\ nE^*m\ \text{and}\ T, m\models \phi,\\
 & \text{and} \ \text{for all}\ n'\in N \ \text{with}\ nE^*n'E^+m \ \text{it holds}\ T, n'\models \psi. 
\end{align*}

\noindent
The $\CTP$ node formula $\until{\phi}{\psi}$ corresponds to the strict until operator, which is expressible in $\ECTL$ as $\exists X\exists \psi U \phi$.

We denote by
$\subseteq_\infty$ the containment  relation for $\ECTL$ node formulas over
finitely branching trees where each branch is infinite.

Let $\phi,\psi$ be in $\CXP^{\neg,\vee}$ and let
 $s$ be a new propositional variable not in $\phi,\psi$. The translation
 $(\cdot)^s$ from $\CXP^{\neg,\vee}$
to $\ECTL$ node formulas relativizes every subformula with $s$ and adjusts the
 syntax.

\begin{eqnarray*}
 (\top)^s &=& s\land \top \\
 (p)^s& = & s\land p \\
 (\neg p)^s &=&s\land \neg p \\
 (\phi_1\land \phi_2)^s & = & \phi^s_1 \land \phi^s_2 \\ 
 (\phi_1\lor \phi_2)^s &=& \phi^s_1\lor \phi^s_2 \\
  (\tup{\downarrow}\phi)^s &=& s \land \exists X \phi^s  \\
(\until{\psi}{\phi})^s &=&  s \land \exists X \exists \psi^sU\phi^s
\end{eqnarray*}
We claim that
 $\phi \subseteq \psi$  iff $\phi^s \subseteq_\infty \psi^s$.

The proof is by contraposition.
First assume that for some finite tree $T=(N, E, r, \rho)$ and node $n\in N$ it holds $T,n\models\phi$ and
$T,n\not\models\psi$. We then construct an infinite tree which is a counterexample. 
Let $s$ be a propositional variable not
occurring in $\phi$ or $\psi$. Let $T^s_\infty$ be the infinite tree
obtained from $T$  by adding to each leaf in $T$  
an infinite path. The labeling is changed as follows: all new nodes
have the empty label set and we add the label $s$ to the label set of
all old nodes.  Formally, $T^s_\infty=(N_\infty, E_\infty, r, \rho_\infty)$, where
$N_\infty=N\cup \{n^m_1,n^m_2, \ldots \mid m \text{ is a leaf in }T\}$ is the set of nodes, $E_\infty=E\cup \{(m, n^m_1)\mid n^m_1\in N_\infty\}\cup \{(n^m_i,n^m_{i+1})\mid n^m_i, n^m_{i+1}\in N_\infty, i\geq 1\}$ is the set of edges, $r$ is the root and $\rho_\infty$ is the labeling function defined as follows. 

\[ \rho_\infty(n) = \left\{ 
  \begin{array}{l l}
    \rho(n)\cup \{s\} & \quad \text{if\ }n\in N, \\
    \emptyset & \quad \text{if\ } n\in N_\infty\setminus N. 
  \end{array} \right.\]

We show by induction on the structure of the formula that for every node $n\in N$ and every $\CTP^{\neg,\vee}$ formula $\theta$, $T,n\models \theta$ iff 
$T^s_\infty, n\models \theta^s$.

To show this we will use the following two facts which are easy consequences of the definitions of $T^s_\infty$ and the translation $(\cdot)^s$:

$$T^s_\infty, n \models s \ \Leftrightarrow \ n\in N, \eqno{(*)}$$
and 
$$T^s_\infty, n\models \phi^s \ \Rightarrow \ T^s_\infty,n \models s. \eqno{(**)}$$

\begin{itemize}
\item $\theta=\top$. We have $T, n \models \top$ iff $n\in N$ iff $n\in N_\infty$ and $s\in \rho_\infty(n)$ iff $T^s_\infty, n \models s\land \top$ iff $T^s_\infty, n\models (\top)^s$.

\item $\theta=p, p\in \Sigma$. We have $T, n\models p$ iff $p\in \rho(n)$ iff (since $n\in N$) $\{p, s\}\subseteq \rho_\infty(n)$ iff $T^s_\infty, n\models s\land p$ iff $T^s_\infty, n\models (p)^s$.

\item $\theta=\neg p, p\in \Sigma$. We have $T, n\models \neg p$ iff $p\not \in \rho(n)$ iff (since $n\in N$ and, thus, $s\in \rho_\infty(n)$) $p\not \in \rho_\infty(n)$ and $s\in \rho_\infty(n)$ iff $T^s_\infty, n\models s\land \neg p$ iff $T^s_\infty, n\models (\neg p)^s$.

\item $\theta=\phi_1\land \phi_2$. We have $T, n\models \phi_1\land \phi_2$ iff $T,n\models \phi_1$ and $T, n\models \phi_2$ iff (by the induction hypothesis,  (*) and  (**)) $T^s_\infty, n\models \phi^s_1$ and $T^s_\infty, n\models \phi^s_2$ iff $T^s_\infty, n\models \phi^s_1\land \phi^s_2$ iff $T^s_\infty, n\models \phi^s$.

\item $\theta=\phi_1\lor \phi_2$. By the same argument as for $\phi_1\land \phi_2$.

\item $\theta=\tup{\downarrow}\phi$. We have $T, n\models \tup{\downarrow} \phi$ iff there exists $m\in N$ such that $nEm$ and $T, m\models \phi$ iff (by the induction hypothesis, (*), (**) and the fact that $n\in N$)  $T^s_\infty, n\models s$ and there exists $m\in N_\infty$ such that $nE_\infty m$ and $T^s_\infty, m\models \phi^s$ iff $T^s_\infty, n\models  s\land \exists X \phi^s$ iff $T^s_\infty \models (\tup{\downarrow}\phi)^s$. 

\item $\theta=\until{\psi}{\phi}$. We have $T, n \models \until{\psi}{\phi}$ iff there exists $m\in N$ such that $T, m\models \phi$ and for all $n'\in N$ with $nE^+ n' E^+ m$ it holds $T, n'\models \psi$. Then by the induction hypothesis, (*) and (**), the latter is equivalent to existence of $m\in N_\infty$ such that $T^s_\infty,m \models \phi^s$ and for all $n'\in N_\infty$ with $nE^+_\infty n' E^+_\infty m$ it holds $T^s_\infty \models \psi^s$. The latter implication holds because for every node $n'$ with $nE^+_\infty n' E^+_\infty m$ it holds $n'\in N$ and $nE^+n'E^+m$. Equivalently, there exists $n_1\in N_\infty$ and $m\in N_\infty$ such that $nE_\infty  n_1E^*_\infty m$, $T^s_\infty, m \models \phi^s$ and for all $n'$ with $n_1E^*_\infty n' E^+_\infty m$ it holds $T^s_\infty, n'\models \psi^s$.  This is equivalent to $T^s_\infty, n \models s\land \exists X\exists \psi^s U \phi^s$ as desired.

\end{itemize}
\noindent Thus, we obtain both $T^s_\infty, n \models \phi^s$ and $T^s_\infty,n \not\models \psi^s$.

For the other direction, let $T_\infty=(N_\infty, E_\infty, r, \rho_\infty)$ be some infinite tree over $\Sigma$ such that $T^\infty,r\models
\phi^s$ and $T^\infty,r\not\models \psi^s$. We then remove from the model all
nodes without $s$ and their descendants.
Formally, $T=(N, E, r, \rho )$, where $N=N_\infty \setminus \{n\mid \exists n' \in N_\infty. s\not\in \rho_\infty(n')\land n'E^*_\infty n\}$ is the set of nodes, $E=E_\infty|_{N\times N}$ is the set of edges, $r$ is the root, and $\rho(n)=\rho_\infty(n)$ is the labeling function.

By induction we can show that  for every $\CTP^{\neg, \vee}$ formula $\theta$ and node $n\in N$,  it holds $T_\infty, n\models \theta^s$ if and only if $T, n\models \theta^s$. 
The two non-trivial cases are the following. 

\begin{itemize}

\item $\theta = \tup{\downarrow}{\phi}$. We have $T_\infty ,n \models s\land \exists X \phi^s$ iff $T_\infty, n\models s$ and there exists $m\in N_\infty$ such that $nE_\infty m$ and $T_\infty, m\models \phi^s$ iff $T, n \models s$ and there exists $m\in N$ such that $nEm$ and $T, m\models \phi^s$. The direction from right to left is obvious, while the direction from left to right because of the following. Since $T_\infty, m\models \phi^s$, it follows that $T_\infty, m \models s$, i.e. $s\in \rho_\infty(m)$. Then since $n\in N$, we have that $m\in N$ too by the definition. The latter also implies that   $nE m$  holds  in $T$
and thus  that $T, n\models  s\land \exists X \phi^s$.

\item $\theta=\until{\psi}{\phi}$. We have $T_\infty , n \models s\land \exists X\exists \psi^s U\phi^s$ iff $T_\infty, n\models s$ and there exists $m\in N_\infty$ such that $T_\infty , m\models \phi^s$, $nE^+_\infty m$ and for all $n'\in N_\infty$ with $nE^+_\infty n'E^+_\infty m$ it holds $T_\infty, n' \models \psi^s$. This is equivalent to $T, n \models s$ and there exists $m\in N$ such that $T, m\models \phi^s$ and for all $n'$ with $nE^{+}n' E^{+} m$ it holds $T, m \models \psi^s$, which means $T, n\models  s\land \exists X\exists \psi^s U\phi^s$ as desired. The direction from left to write is trivial since $T$ is a substructure of $T_\infty$. The direction from right to left follows from the fact that $m\in N$ and all the nodes $n'$ such that $nE^+_\infty n' E^+_\infty m$ belong to $N$ as well and, moreover, it holds $nE^{+} n' E^{+} m$. 

\end{itemize}

\noindent Thus, we obtain that $T, r\models \phi^s$ and $T, r\not\models \psi^s$. Because all nodes of $T$ now make $s$
true, we can discard the relativization with $s$ in the formulas. Thus obtaining $T,r \models \phi$ and $T, r\not\models \psi$. 
The only problem is that $T$ is infinite.
 But using the simulations from Theorem~\ref{prop:simc} it is easy to see that  we can turn $T$ into a finite counterexample.

Thus, we have shown that $\phi\subseteq \psi$ if and only if $\phi^s\subseteq_\infty \psi^s$. The latter containment problem in $\ECTL$ is known to be decidable in $\PSPACE$. This concludes the proof.
\end{proof}

\begin{cor}
The containment problem for $\TP^{\neg}$ is  $\PSPACE$ complete. 
\end{cor}

However, restricting negation to a safe negation $p\land \neg q_1\land \ldots \land \neg q_n$  keeps the complexity of the containment problem for tree patterns in $\coNP$. Note that  $p\land \neg q_1\land \ldots \land \neg q_n$ only adds expressive power over multi-labeled models, because  on single labeled models $p\equiv p\land\neg q_1\land \ldots \land \neg q_n$. The proof of the following result can be found in \cite{marx:cont14}. Here $\TP^{\lor, \neg^s}$ is tree patterns expanded with disjunction and the safe negation construct. 

\begin{thm}
The containment problem for $\TP^{\lor, \neg^s}$ over multi-labeled trees is in $\coNP$. 
\end{thm}

\section{Conclusion} We have shown that adding conditions on the edges
of tree patterns gives a boost in expressive power which comes with
the price of a higher -- $\PSPACE$ -- complexity for the containment
problem than for tree patterns.
We located the source of the extra complexity in the fact that
unrestricted negations of labels can be coded  in Conditional Tree
Patterns. Adding negations of labels to tree patterns causes an increase in complexity of the containment problem from $\coNP$ to $\PSPACE$.

This paper is a first step in exploring ``Regular Tree Patterns'' and
fragments of it. We mention some directions for future work.

Miklau and Suciu in \cite{mikl:cont04} mention that existence of a
homomorphism between tree patterns is a necessary but not  sufficient condition for
containment in $\TP$. Can we extend the simulations between
conditional tree patterns and trees to simulations between queries,
partly capturing containment as for tree patterns?

What is the complexity of containment of regular tree patterns,
i.e. the positive fragment of Regular XPath without disjunction and union? 
As the satisfiability (and thus also the containment) problem for
Regular XPath is known to be  $\EXPTIME$-complete, it must lie in
between $\PSPACE$ and $\EXPTIME$. 

There are also interesting characterization questions: what
is the exact $\mathbf{FO}$ fragment corresponding to $\CXP$ or $\CXP$ with
disjunction and union?  Ten Cate~\cite{cate:expr06} showed that full Regular
XPath (with all four axis relations) expanded with path equalities is equally expressive  as
binary $\mathbf{FO}^*$ (First-order logic extended with a transitive closure
operator that can be applied to formulas with exactly two free
variables). Is every positive forward binary $\mathbf{FO}^*$ formula
expressible as a Regular Tree Pattern? Is
$\CXP$ with disjunction and union equally expressive as  $\mathbf{FO}$
intersected with positive $\mathbf{FO}^*$?  Are unions of $\CXP$ equivalent to
unions of the first order fragment of conjunctive regular path
queries~\cite{calv:cont00}?

\section{Acknowledgement}
We would like to thank the anonymous referees for their helpful comments.

\bibliographystyle{abbrv}
\bibliography{main-bib}

\begin{thebibliography}{10}

\bibitem{amer:tree02}
S.~Amer-Yahia, S.~Cho, L.~Lakshmanan, and D.~Srivastava.
\newblock Tree pattern query minimization.
\newblock {\em The VLDB Journal}, 11:315--331, 2002.

\bibitem{bara:guar11}
V.~B{\'a}r{\'a}ny, B.~ten Cate, and L.~Segoufin.
\newblock Guarded negation.
\newblock In L.~Aceto, M.~Henzinger, and J.~Sgall, editors, {\em ICALP (2)},
  volume 6756 of {\em Lecture Notes in Computer Science}, pages 356--367.
  Springer, 2011.

\bibitem{bene:stru05}
M.~Benedikt, W.~Fan, and G.~M. Kuper.
\newblock Structural properties of {XPath} fragments.
\newblock {\em Theor. Comput. Sci.}, 336(1):3--31, 2005.

\bibitem{bjor:conj11}
H.~Bj{\"o}rklund, W.~Martens, and T.~Schwentick.
\newblock Conjunctive query containment over trees.
\newblock {\em J. Comput. Syst. Sci.}, 77(3):450--472, 2011.

\bibitem{blac:moda01}
P.~Blackburn, M.~{de Rijke}, and Y.~Venema.
\newblock {\em Modal Logic}.
\newblock Cambridge University Press, 2001.

\bibitem{calv:cont00}
D.~Calvanese, G.~D. Giacomo, M.~Lenzerini, and M.~Y. Vardi.
\newblock Containment of conjunctive regular path queries with inverse.
\newblock In {\em KR}, pages 176--185, 2000.

\bibitem{chle:domi86}
B.~S. Chlebus.
\newblock Domino-tiling games.
\newblock {\em J. Comput. Syst. Sci.}, 32(3):374--392, 1986.

\bibitem{clar:auto86}
E.~M. Clarke, E.~A. Emerson, and A.~P. Sistla.
\newblock Automatic verification of finite-state concurrent systems using
  temporal logic specifications.
\newblock {\em ACM Trans. Prog. Lang. Syst.}, 8:244--263, 1986.

\bibitem{gott:effi05}
G.~Gottlob, C.~Koch, and R.~Pichler.
\newblock Efficient algorithms for processing {XPath} queries.
\newblock {\em ACM Trans. Database Syst.}, 30(2):444--491, 2005.

\bibitem{kime:revi08}
B.~Kimelfeld and Y.~Sagiv.
\newblock Revisiting redundancy and minimization in an {XPath} fragment.
\newblock In {\em EDBT'08}, pages 61--72, 2008.

\bibitem{kupf:auto00}
O.~Kupferman and M.~Y. Vardi.
\newblock An automata-theoretic approach to modular model checking.
\newblock {\em ACM Trans. Prog. Lang. Syst.}, 22(1):87--128, 2000.

\bibitem{libk:reas10}
L.~Libkin and C.~Sirangelo.
\newblock Reasoning about {XML} with temporal logics and automata.
\newblock {\em Journal of Applied Logic}, 8(2):210 -- 232, 2010.
\newblock Selected papers from the Logic in Databases Workshop 2008.

\bibitem{marx:cond05}
M.~Marx.
\newblock Conditional {XPath}.
\newblock {\em ACM Trans. Database Syst.}, 30(4):929--959, 2005.

\bibitem{marx:cont14}
M.~Marx and E.Sherkhonov.
\newblock Containment for queries over trees with attribute value comparisons.
\newblock {\em Submitted}, 2015.

\bibitem{mikl:cont04}
G.~Miklau and D.~Suciu.
\newblock Containment and equivalence for a fragment of {XPath}.
\newblock {\em J. ACM}, 51(1):2--45, 2004.

\bibitem{neve:comp06}
F.~Neven and T.~Schwentick.
\newblock On the complexity of {XPath} containment in the presence of
  disjunction, {DTDs}, and variables.
\newblock {\em Logical Methods in Computer Science}, 2(3), 2006.

\bibitem{cate:expr06}
B.~ten Cate.
\newblock The expressivity of {XPath} with transitive closure.
\newblock In S.~Vansummeren, editor, {\em PODS}, pages 328--337. ACM, 2006.

\bibitem{emde:conv96}
P.~{van Emde Boas}.
\newblock The convenience of tilings.
\newblock In A.~Sorbi, editor, {\em Complexity, Logic and Recursion Theory},
  volume 187 of {\em Lecture notes in pure and applied mathematics}, pages
  331--363. Marcel Dekker Inc., 1997.

\end{thebibliography}
\appendix

\section{Translations between $\CTP$ and $ctp$}
\label{app:tra}

The following constructions are similar to the tree and graph representations for tree patterns and conjunctive queries over trees~\cite{mikl:cont04, bjor:conj11}

We define the translation function $\ct(\cdot)$ which assigns an equivalent  conditional tree pattern with one output node to  every $\CTP$ formula. The output node of $\ct(\phi)$, where $\phi$ is a node formula, equals the root.

Let $\alpha$ and $\phi$ be path and node $\CTP$ formulas. 
We define $\ct(\alpha) $ and $\ct(\phi)$ by mutual induction on the complexity of the path and node formulas. We take $\ct(\alpha)$ or $\ct(\phi)$ to be the tree $ (N, E , r, o, \rho_N, \rho_E)$, where the components are defined according to the cases. 

\begin{itemize}

\item $\alpha = \downarrow$. Then $N$ consists of two nodes $r$ and $n$. The edge relation $E$ is defined as $ \{\tup{r, n}\}$. Moreover, $o := n$, $\rho_N(v) = \emptyset$ for $v\in \{r, n\}$ and $\rho_E(\tup{r,n}) = \downarrow$,

\item $\alpha = ? \phi$. Then $\ct(\alpha) := \ct(\phi)$,

\item $\alpha = \alpha_1 ; \alpha_2$. Let $\ct(\alpha_1)$ and $\ct(\alpha_2)$ be the conditional tree patterns for $\alpha_1$ and $\alpha_2$. Then $c(\alpha)$ is the tree obtained as follows. We fuse the root of $\ct(\alpha_2)$ with the output node of $\ct(\alpha_1)$ and declare the output node of $\ct(\alpha_2)$ as the output node of $\ct(\alpha)$. The label of the fusion node is   the union of the labels of the output node of $\ct(\alpha_1)$ and the root of $\ct(\alpha_2)$. Labels of other nodes and the edges   remain the same as in $\ct(\alpha_1)$ and $\ct(\alpha_2)$. 

\item $\alpha = \untilpath$. Then $N$ consist of two nodes $r$ and $n$. The edge relation $E$ is defined as $ \{\tup{r, n}\}$. Moreover, $o := n$, $\rho_N(v) = \emptyset$ for $v\in \{r, n\}$ and $\rho_E(\tup{r,n}) = \ct(\phi)$.

\end{itemize}

\noindent For node formulas, the  output node of the translation result is   always defined as   the root $r$. 

\begin{itemize}

\item $\phi = p$. Then $N$ consists of a single node $r$, $E$ is empty,  the labeling $\rho_N(r)=\{p\}$.

\item $\phi = \top$. Similar to the previous case, with the exception that $\rho_N(r) = \emptyset$

\item $\phi = \phi_1 \wedge \phi_2$. Then $\ct(\phi) := \ct(\phi_1)\oplus \ct(\phi_2)$, i.e. the fusion of the conditional tree patterns $\ct(\phi_1)$ and $ \ct(\phi_2)$.

\item $\phi= \tup{\alpha} \phi_1$. Let $\ct(\alpha)$ and $\ct(\phi_1)$ be the corresponding conditional tree patterns. Then $\ct(\phi)$ is obtained by fusing the output node of $\ct(\alpha)$ with the root of $\ct(\phi)$. The labeling of the fusion node is defined as the union of the labels of the root  of  $\ct(\phi)$ and the output node of  $\ct(\alpha)$. Labels of the other nodes and edges   remain the same as in $\ct(\alpha)$ and $\ct(\phi)$.

\end{itemize}

\noindent Translation  $f(\cdot)$  works the other way around. Let $t= (N, E , r, o, \rho_N, \rho_E)$ be a conditional tree pattern with one output node. We define   $f(t)$ by induction on the nesting depth and the depth of the tree. We first define a mapping $\phi$ from nodes $v\in N$ to $\CTP$ node formulas. The mapping is defined inductively starting from the leaves as follows. Here, for a finite set $S= \{p_1, \ldots,  p_n\}$, by $\wedge S$ we denote the finite conjunction $p_1\wedge \ldots \wedge p_n$. We take $\wedge \emptyset$ to be $\top$.
For a conditional tree pattern $t$, by $r_t$ we denote the root of $t$.
For $v$ a node in pattern $t$:
\begin{align*}
\phi(v) &=\wedge \rho_N(v)\  \wedge \\ 
& \bigwedge_{\tup{v, v'}\in E ,\  \rho_E(v, v')=t'} \tup{? \phi(r_{t'})}\phi(v')\  \wedge \\
& \bigwedge_{\tup{v, v'} \in E , \ \rho_E(v, v') = \downarrow} \child{\phi(v')}
\end{align*}
Now let $r=v_1, \ldots, v_n=o$ be the path from the root $r$ to the output node $o$ in $t$. Let the expression $d_i, 1\leq i\leq n-1$ be defined by: $d_i = \downarrow$ if $\rho_E(v_i, v_{i+1}) =\downarrow$ and $d_i = \upath{\phi(r_{t'})}$ if $\rho_E(v_i, v_{i+1}) = t'$. Then the result of the translation $f(t)$ of the conditional tree pattern $t$ is the $\CTP$ path formula:

\[
?\phi(v_1); d_1 ; ?\phi(v_2); d_2; \ldots; ?\phi(v_{n-1}); d_{n-1} ;? \phi(v_n).
\]

For the next proposition we need the definition of  equivalence between conditional tree patterns. Let $t$ be a conditional tree pattern with output nodes $\bar{o}$,$ |\bar{o}|$=$k$, and $T$ a tree. Then the \emph{answer set} of $t$ over $T$ is the set $Out(t, T)=\{\tup{g(o_1), \ldots, g(o_k)}\mid g \text{ is a simulation of } t \text{ in } T\}$. We say that two conditional tree patterns $t_1$ and $t_2$ are \emph{equivalent}, denoted as $t_1\simeq t_2$, if $Out(t_1, T) = Out(t_2, T)$ for every tree $T$.

\begin{prop}
Let $\phi$ and $\alpha$ be $\CTP$ node and path formulas, $t$ a conditional tree pattern. Then it holds that 
\begin{enumerate}[label=(\roman*)]
\item $f(\ct(\phi)) \equiv \phi$ and $f(\ct(\alpha)) \equiv \alpha$,
\item $\ct(f(t)) \simeq t$.
\end{enumerate}
\end{prop}

\noindent The proof of (i) is by mutual induction on $\alpha$ and $\phi$, and the proof of (ii) is by induction on nesting depth and depth of $t$.

\end{document}